\documentclass[conference]{IEEEtran}
\ifCLASSINFOpdf
\else
\fi
\usepackage{amsthm} 
\usepackage{amssymb}  
\usepackage{graphicx}
\usepackage{float} 
\usepackage{wrapfig} 
\usepackage{subcaption}
\usepackage{physics} 
\usepackage{graphicx}      
\usepackage[table]{xcolor} 
\usepackage{array}         
\usepackage{float}         
\usepackage{enumitem}

\usepackage{amsmath,amsfonts}
\usepackage{algorithmic}
\usepackage{array}
\usepackage{textcomp}
\usepackage{stfloats}
\usepackage{url}
\usepackage{verbatim}
\usepackage{graphicx}
\usepackage{balance}
\usepackage{xcolor}
\usepackage{blkarray}
\usepackage[hidelinks]{hyperref}

\DeclareUnicodeCharacter{2212}{-}
\graphicspath{{img}}

\theoremstyle{colon}

\theoremstyle{definition}
\newtheorem{proposition}{Proposition}

\newtheorem{lemma}{Lemma}
\newtheorem{remark}{Remark}
\newcommand{\I}{\boldsymbol I}
\newcommand{\0}{\boldsymbol 0}

\begin{document}
%

\title{Formation Matrix and Energy-based Control of Multi-Agent Systems}

\author{\IEEEauthorblockN{Martín Crespo and Sergio Junco and Matías Nacusse}
\IEEEauthorblockA{Automation and Control Systems Laboratory - FCEIA,\\ National University of Rosario, Argentine\\
Email: \{crespom, sjunco, nacusse\}@fceia.unr.edu.ar}}

\maketitle

\begin{abstract}
This paper presents an energy-based controller for a multiagent robotic system designed to achieve and maintain a specific formation while moving on a plane and avoiding collisions between agents. 
The controller emulates a network of elementary spring-damper modules connecting pairs of agents. 
This network, with its de-energized states representing the desired formation, determines the system's dynamics, which is fully encapsulated by a bond graph model. 
The modeling is further enhanced through the introduction of a formation matrix, using a graph-theoretic approach, that describes both  the distances and relative velocities among the agents of the arrangement. 
This matrix mathematically represents the interconnection and energy-exchange structure of the bond graph, allowing us to put it in correspondence with the control-by-interconnection CbI-scheme of the IDA-PBC theory, facilitating the solution of the formation control problem within the port-Hamiltonian system framework. 
Furthermore, the paper presents leader-following and position-based formation control systems based on the CbI scheme, including a stability analysis of the corresponding closed-loop systems. The theoretical findings are validated through numerical simulations across various scenarios.
\end{abstract}

\section{Introduction}



The study of multi-agent systems (MAS) has been increasingly capturing the attention of the scientific community due to their ability to collaboratively solve complex problems by breaking them down into smaller tasks. This is also thanks to their diverse range of applications, spanning from robotics and micro smart grids to telecommunications, logistics, transportation and search and rescue operations \cite{queralta2020collaborative}.
In the case of robotics MAS, the formation control problem is one of the most actively studied topics \cite{yang2021attacks}.


Graphs have been extensively used to model the interaction topology of MAS, by associating mobile robots (or agents) with the nodes and representing either the Euclidean distances or the existence of interactions between them through the edges of the graph \cite{oh2015survey}.
Several works have also considered the bearings of a graph to model the topology of a MAS and to address the problem of formation control using an angle-based approach \cite{jing2019angle} \cite{chen2020angle}.


In this paper, the edges of the graph represent the Euclidean distances between agents. 
Then we introduce the concept of \textit{Formation Matrix}, which collects the cosines and sines of the rotation angles of the edges and enables the representation of both distances and relative velocities between agents.



Our objective in this article is to provide a physically, heuristically inspired, yet theoretically supported approach to the design and analysis of the formation control problem for a MAS composed of $n$ mobile wheeled omnidirectional robots.
The work \cite{khatib1986real} was among the pioneers in employing a virtual physics-based design, introducing the concept of artificial potential fields.
Several works have considered interconnection structures based on virtual springs and dampers \cite{chen2015virtual} \cite{chen2016shape} \cite{wiech2018virtual} \cite{piemngam2019virtual}.


Among the different classifications of formation control proposed by \cite{scharf2004survey}, we recall the behavioral approach —which includes collision avoidance— and the virtual structure approach, in which the formation of agents is treated as a virtual structure.
In line with these approaches, we propose an Interconnected Physical Structure (IPS) in which the agents are coupled to each other via sets of virtual spring-damper couplings, referred here to as Elementary Coupling Block (ECB). This IPS allows for controlling the distances between agents while simultaneously avoiding collisions among them.



In this context, it is reasonable to consider, both in modeling and control, the energy phenomena associated with the physical-based laws, such as its (i) storage, (ii) flow, and (iii) dissipation.
When approaching the control problem in energy-related terms, its design can be addressed via the passivity-based control (PBC) techniques \cite{ortega2001putting} applied to the feedback-interconnection closed-loop structure.
Results on PBC of MAS are presented in \cite{li2022passivity}.


In terms of modeling techniques, the Euler-Lagrange (EL) and port-Hamiltonian Systems (pHS) models are natural candidates for describing numerous physical systems, as they capture the energy phenomena outlined above \cite{ortega2013passivityChapter1} \cite{van2000l2}.
We also consider the Bond Graphs (BG), a graphical representation that allows capturing energy phenomena as well as interconnection structure and causal
relationships between system variables, and is a highly suitable tool for both modeling and controller design \cite{karnopp2012system}. 
Indeed, controller design methods based on BG provide a powerful framework, as they enable a feasible implementation of controllers inspired by physical principles in complex systems.
In \cite{vos2014formation}, a solution to the formation control problem for a network of mobile robots with differential locomotion is presented, employing pHS models and virtual couplings.



In this paper, we specifically explore the PBC technique of Control by Interconnection (CbI), where the controller is connected to the plant through a power-preserving interconnection \cite{ortega2004interconnection}. 
BG models are used to represent the IPS and serve as a \textit{lingua franca} between the port-Hamiltonian (pH) models and CbI.

In this work, the controller state variables of the closed-loop system are related to the plant states via Casimir invariant sets \cite{garcia2005control}.
We demonstrate that this is a property produced by the geometry of the IPS. 
Casimir functions are conserved quantities of the system for any choice of the Hamiltonian, and so are completely determined by the geometry (i.e., the interconnection structure) of the system \cite{ortega2001putting}. 
Finally, as the kinematic constraints are holonomic \cite{siciliano2009robotics}, the controller can be represented by a static feedback scheme in terms of the plant states.




In comparison to the aforementioned studies, the novelty of this work is given by the introduction of the concept of the \textit{Formation Matrix} in the formulation of the control problem. 
Unlike previous approaches such as \cite{jing2019angle} and \cite{chen2020angle}, which focus solely on angle-based representations, our method simultaneously incorporates both inter-agent distances and relative velocities between agents. 
While the virtual coupling strategies in  \cite{khatib1986real}-\cite{piemngam2019virtual} adopt physically inspired models, they often lack a systematic energy-based framework. 
In contrast, this work leverages BG modeling to integrate pH representations with the CbI technique, providing a structured PBC architecture in which the \textit{Formation Matrix} plays a key role in modulating the power exchange between the ECBs and the agents.
Furthermore, unlike \cite{chopra2006passivity} and \cite{vos2014formation},  where energy-based models are decoupled from interaction geometry, our proposed IPS inherently leads to Casimir invariants sets, enabling a static feedback interpretation of the MAS due to its geometric design and holonomic constraints. This integration results in a more comprehensive and physically consistent approach to the design and analysis of MAS formation control.


The remainder of the paper is organized as follows. 
In Section \ref{sec:preli}, we provide some preliminaries on pHS, graph theory and agent models. We also establish the notation and formally state the problem. 
In Section \ref{sec:formationmatrix}, we define the \textit{Formation Matrix} and demonstrate some of its properties. 
Section \ref{sec:phy_strucuture} addresses the models of both the ECB and the MAS, initially modeling the agents as point masses, and later extending the results to rigid-body agents models.
In Section \ref{sec:FC}, some properties related to the CbI and the \textit{Casimir} functions are presented. 
In Section \ref{sec:LAC} and Section \ref{sec:poscontrol}, two different formation control strategies are presented using the proposed IPS. A stability analysis is included, also showing that the control law remains valid even when modeling the robots as rigid bodies instead of point masses.
Finally, Section \ref{sec:sim} showcases simulation results and Section \ref{sec:Con} addresses concluding remarks and discusses future research.

\section{Preliminaries}\label{sec:preli}

This section introduces the notation, an overview of the pHS and graph theory, the model of the agents to be used, and the problem statement. These concepts will be used throughout the remainder of the paper.

\subsection{Notation}

We denote the $n$-dimensional identity matrix by $\I_n$ and $\0_{n\times m}$ represents a matrix of zeros of size $n$ by $m$. 
$A=\text{diag}\{[a_1,\cdots, a_n]\}$ denotes the block diagonal matrix such that $a_n$ is its $n$-th diagonal element. 
For a scalar function $H(x):\mathbb{R}^n\rightarrow \mathbb{R}$, $\nabla H=\frac{\partial H}{\partial x}$ denotes the column vector of partial
derivatives of the function $H(x)$ with respect to $x$.
Given two matrices $A$ and $B$, $A\otimes B$ denotes the Kronecker product of the matrices.

\subsection{Port-Hamiltonian systems}
The port-Hamiltonian (pH) dynamics is given by
\begin{equation*}
\Sigma_{(u,y)}:\begin{cases}
\dot{x}&=\left[J(x)-R(x)\right]\nabla H(x)+g(x)u\\
y&=g^T(x)\nabla H(x)
\end{cases}
\end{equation*}
where $x\in\mathbb{R}^n$ is the state vector, 
$u\in\mathbb{R}^m$ is the input, with $m\leq n$, 
the Hamiltonian $H(x):\mathbb{R}^n\rightarrow \mathbb{R}$ is the total stored energy, 
$J(x):\mathbb{R}^n\rightarrow \mathbb{R}^{n\times n}$ is the skew-symmetric interconnection matrix, 
$R(x):\mathbb{R}^n\rightarrow \mathbb{R}^{n\times n}$ the symmetric and positive semi–definite dissipation matrix, 
$y\in\mathbb{R}^m$ is the natural output as $u^Ty$ has units of power and $g(x):\mathbb{R}^n\rightarrow \mathbb{R}^{n\times m}$ is assumed full rank \cite{van2016l2}.

\subsection{Graph Theory}
A graph $\mathcal{G(V,E)}$ is an ordered pair with $\mathcal{V}=\{1,2,\ldots,n \}$ distinct nodes and $\mathcal{E}\subset\mathcal{V}\times\mathcal{V}$ the set of edges, being $|\mathcal{E}|=m$. A directed edge $(i,j)\in\mathcal{E}$ is an ordered pair of two distinct nodes. A graph $\mathcal{G(V,E)}$ is directed if consists of directed edges. 
In a graph $\mathcal{G}$ a node $i$ is called a neighbor of node $j$ if the link $(i,j)$ exists in the graph $\mathcal{G}$ \cite{godsil2001algebraic}.

One of the several matrices that can be associated with a graph is the incidence matrix. Let $\mathcal{G(V,E)}$ be a directed graph. Assigning arbitrary orientation to the edges, considering one of the nodes of an edge as the positive end and the other as the negative end, the incidence matrix $D \in \mathbb{R}^{n\times m}$ is defined as
\begin{equation*}
D_{ik}=\begin{cases}
-1 \quad \textrm{if node i is the positive end of edge k}\\
1 \quad \textrm{if node i is the negative end of edge k}\\
0\quad \textrm{otherwise}
\end{cases}
\end{equation*}

Let's consider a set of robots modeled as point masses deployed on a horizontal plane. Before addressing their interaction, we use graph theory to describe some geometrical aspects of this arrangement. We assign the nodes and edges the following attributes: each node is assigned the Cartesian position of a robot, and each edge is assigned a length defined as the distance between a pair of agents.
More formally, assume that  $i$ and $j$ are two neighbor nodes, being $i$ the positive end of edge $k$. 
Let $r_i={\left[ x_i,\;y_i\right]}^T\in \mathbb{R}^2$ and $r_j={\left[ x_j,\;y_j\right]}^T\in \mathbb{R}^2$ be the position of the nodes $i$ and $j$ respect to a fixed reference frame, respectively.
The displacement $\Delta r_k={\left[ \Delta r_{xk},\;\Delta r_{yk}\right]}^T\in\mathbb{R}^2$ with respect to a fixed reference frame is $\Delta r_k=r_i-r_j$.
In this way, from the definition of $D$, the displacement between all the neighboring nodes $\Delta r={\left[\Delta r_1^T,\ldots,\Delta r_m^T \right]}^T\in\mathbb{R}^{2m}$ can be expressed as 
\begin{equation*}
\Delta r=\left(D^T \otimes\I_2 \right)r
\end{equation*}
where $r=[r_1^T,\cdots, r_n^T]^T\in\mathbb{R}^{2n}$. 
Finally, the Euclidean distance $z_k$ between nodes $i$ and $j$ can be written as
\begin{equation}\label{eq:distance}
z_k=\sqrt{\Delta r_k^T\Delta r_k}
\end{equation}

\subsection{Mobile Base Models}

In this section, we present the two types of models that we will use in this work.

\subsubsection{Point Mass (State-equation) Model}
The dynamics of agent $i$ is
\begin{equation*}
\begin{cases}
\dot{r}_i&=M_i^{-1}p_i\\
\dot{p}_i&=u_i
\end{cases}
\end{equation*}
where 
$p_i=[p_{ix},\; p_{iy}]^T\in \mathbb{R}^2$ denotes the momentum, $M_i=\text{diag}\{[m_i,\; m_i]\}$, with $m_i$ the mass of the agent, is its mass matrix, and $u_i=[u_{ix},\; u_{iy}]^T\in \mathbb{R}^2$ is the control input of agent $i$.
Thus, the dynamics of $n$ agents is expressed as
\begin{equation}\label{eq:pto_mat_din}
\begin{cases}
\dot{r}&=M^{-1}p\\
\dot{p}&=u
\end{cases}
\end{equation}
where $p=[p_1^T,\cdots, p_n^T]^T\in\mathbb{R}^{2n}$, $u=[u_1^T,\cdots, u_n^T]^T\in\mathbb{R}^{2n}$, and $M=\text{diag}\{[M_1,\cdots, M_n]\}$.

\subsubsection{Rigid Body (Euler-Lagrange) Model}

To fully characterize the configuration of the omnidirectional mobile base $i$ in the plane, we must extend the point mass model to a rigid body model specifying not only its position $r_i$ but also its rotation angle $\varphi_i$, so that the generalized coordinate vector of the rigid body model is $q_i={\left[ x_i,\;y_i,\;\varphi_i\right]}^T\in \mathbb{R}^3$.

The model of a single rigid-body platform, say $i$, enhanced with the model of its locomotion system, composed of three omni-wheels symmetrically deployed at 120° each other, reads as follows
\begin{equation*}
M_{bi}\ddot{q_i}+C_i(q_i,\dot{q}_i)\dot{q}_i= G_i(q_i)\tau_i
\end{equation*}
where 
$C_i(q_i,\dot{q}_i)$ is the so-called Coriolis matrix, whose product with $\dot{q}_i$ represents the centrifugal and Coriolis efforts present in the model, and 
$G_i(q_i)$ is the matrix relating the wheel torques $\tau_i={\left[ \tau_{r1i},\;\tau_{r2i},\;\tau_{r3i}\right]}^T\in \mathbb{R}^3$ to efforts at the center of mass of the base with expressions:
\begin{equation*}\label{eq:MatrizMmatrizC}
\begin{gathered}
C_i(q_i,\dot{q}_i)=\begin{bmatrix}
       0 & \frac{3J_{r_i}}{2R_i^2}\dot{\varphi}_i &0 \\[0.3em]
       -\frac{3J_{r_i}}{2R_i^2}\dot{\varphi}_i & 0 & 0 \\[0.3em]
       0 & 0 & 0
\end{bmatrix}\\
G_i(q_i)=\frac{1}{R_i}\begin{bmatrix}
       c_{\varphi_i+\frac{\pi}{3}} & -c_{\varphi_i} & c_{\varphi_i-\frac{\pi}{3}} \\[0.3em]
       s_{\varphi_i+\frac{\pi}{3}} & -s_{\varphi_i} & -c_{\varphi_i+\frac{\pi}{6}} \\[0.3em]
       L_i & L_i & L_i
\end{bmatrix}
\end{gathered}
\end{equation*} 
and $M_{bi}=\text{diag}\{[\frac{3J_{r_i}}{2R_i^2}+m_{b_i},\frac{3J_{r_i}}{2R_i^2}+m_{b_i}, \frac{3J_{r_i}L_i^2}{2R_i^2}+I_{b_i}]\}$ is the symmetric and positive definite inertia matrix,
where $R_i$ is the radius of the wheels,
$L_i$ is the distance from the center of mass of the chassis to the wheel axis,
$J_{r_i}$ is the rotational inertia of the wheels, and
$m_{b_i}$ and $I_{b_i}$ are the mass and rotational inertia of the mobile base of agent $i$, respectively.

The EL model of $n$ agents characterized as mobile bases is 
\begin{equation}\label{eq:MB_dynamic_equation}
M_b\ddot{q}+C(q,\dot{q})\dot{q}= G(q)\tau
\end{equation}
where 
$M_b=\text{diag}\{[M_{b1}, \; \cdots ,\;M_{bn}]\}$,
$q=\left[q_1^T,\cdots,q_n^T\right]^T\in\mathbb{R}^{3n}$, 
$C(q,\dot{q})=\text{diag}\{[C_1(q_1,\dot{q}_1), \; \cdots ,\;C_n(q_n,\dot{q}_n)]\}$,
$G(q)=\text{diag}\{[G_1(q_1), \; \cdots ,\;G_n(q_n)]\}$, and
$\tau=\left[\tau_1^T,\cdots,\tau_n^T\right]^T\in\mathbb{R}^{3n}$.

\subsection{Problem Statement and Solution Approach}

In this paper we tackle two objectives: (i) the distance control and (ii) the collision avoidance between agents.
For that we propose a physically inspired strategy, which emulates the interconnection among the $n$ agents using $m$ sets of virtual springs and dampers, a collection of single sets that we define as Elementary Coupling Block (ECB), resembling the IPS shown in Fig.~\ref{fig:IF}.

The mesh of ECB's enables the achievement of both proposed objectives.
On one hand, since the forces of the springs vanish when they are de-energized, selecting their lengths allows establishing the desired distances between agents.  
On the other hand, as the spring forces are set to tend to infinity when the distances approach zero, they serve as a collision avoidance strategy.
Finally, when the distance between agents exceeds the desired value, the springs apply a force to attract them. The dampers, for their part, attenuate the oscillations of the mesh.

A virtual coupling should emulate what a real physical coupling would produce if it were present. 
To implement these virtual couplings, the efforts that the physical couplings would generate are translated into force references for the agent actuators (the wheel torques generated by the motors) which are responsible for producing the actual efforts.

\begin{figure}[b!]
\centering
\includegraphics[width=\columnwidth]{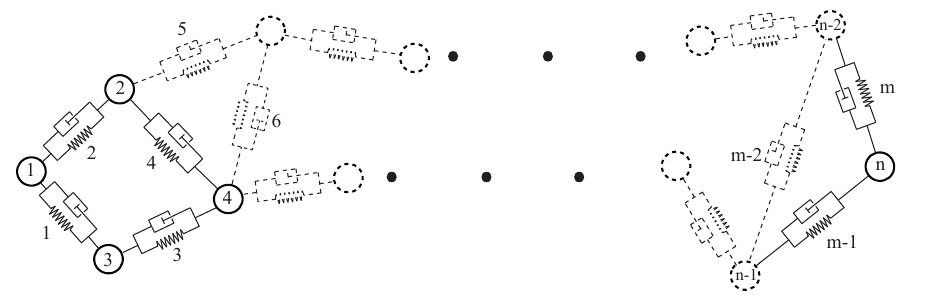}
\caption{\label{fig:IF}IPS with $n$ agents and $m$ ECBs.}
\end{figure}

\section{Formation Matrix}\label{sec:formationmatrix}

In this section, we introduce the \textit{Formation Matrix} concept, which will be crucial for modeling the IPS of Fig. \ref{fig:IF}. 
It serves to compactly express both the dependence of the edges' lengths on the robots' positions and the rates of change of the lengths with respect to the robot velocities, i.e., the kinematic and differential kinematic constraints of the setup.

Let's consider any two neighbor nodes in a graph $\mathcal{G}$, as shown in Fig. \ref{fig:Rotation}.
The euclidean distance defined in (\ref{eq:distance}) can be rewritten as
\begin{equation*}
\begin{split}
z_k&=\begin{bmatrix}
       \cos\alpha_k,\;\sin\alpha_k\\
\end{bmatrix}\Delta r_k
\end{split}
\end{equation*}
where
\begin{equation*}
\cos\alpha_k=\dfrac{\Delta r_{xk}}{z_k} \qquad \sin\alpha_k=\dfrac{\Delta r_{yk}}{z_k}
\end{equation*}
and $\alpha_k$ is the angle formed by the edge $k$ with the $x$-axis of the fixed reference frame. Note that the cosine and sine functions can be expressed as functions of the positions $r$, e.g., the cosine of edge $k=1$ in Fig. \ref{fig:IF} is:
\begin{equation*}
\begin{split}
\cos\alpha_1&=\frac{\Delta r_{x1}}{\sqrt{\Delta r_1^T\Delta r_1}}\\
&=\frac{x_3-x_1}{\sqrt{(r_3-r_1)^T(r_3-r_1)}}
\end{split}
\end{equation*}

\begin{figure}[t!]
\begin{centering}
\includegraphics[width=0.7\columnwidth]{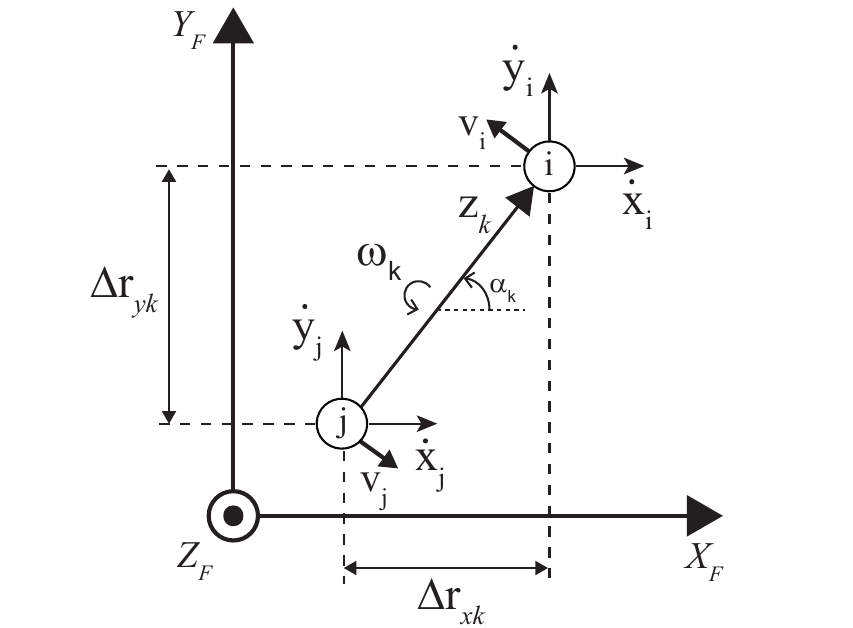}
\par\end{centering}
\caption{\label{fig:Rotation}Neighbor nodes $i$ and $j$ of graph $\mathcal{G}$.}
\end{figure}

In this way, the vector of distances between all the neighboring nodes of a graph can be expressed as
\begin{equation}\label{eq:distance_2}
z=\mathcal{F}(r)r
\end{equation}
where $z={\left[z_1,\cdots,z_m\right]}^T\in\mathbb{R}^{m}$ and $\mathcal{F}(r)$ is the \textit{Formation Matrix} defined as
\begin{equation}\label{eq:F}
\mathcal{F}(r):=\beta(r)^T\left(D^T\otimes \I_2 \right)
\end{equation}
being%
\footnote{The notations $c_k$ and $s_k$ are the abbreviations for $\cos\alpha_k$ and $\sin\alpha_k$, respectively.}
\begin{equation*}
\begin{split}
\beta(r)&=\begin{bmatrix}
       c_1 & 0 & \cdots & 0 \\[0.3em]
	   s_1 & 0 & \cdots & 0 \\[0.3em]
	   0 & c_2 & \cdots & 0 \\[0.3em]
	   0 & s_2 & \cdots & 0 \\[0.3em]
	   \vdots & \vdots & \vdots & \vdots  \\[0.3em]
	   0 & 0 & 0 & c_m \\[0.3em]
	   0 & 0 & 0 & s_m \\[0.3em]
     \end{bmatrix} 
\end{split} 
\end{equation*}
a matrix of dimension $({2m\times m})$.

On the other hand, the time derivative of (\ref{eq:distance}) is
\begin{equation*}
\begin{split}
\dot{z}_k&=\dfrac{\Delta r_{xk}\dot{\Delta} r_{xk}+\Delta r_{yk}\dot{\Delta} r_{yk}}{\sqrt{\Delta r_{xk}^2+\Delta r_{yk}^2}}\\
&=\begin{bmatrix}
       c_k,\;s_k\\
\end{bmatrix}
\dot{\Delta} r_k
\end{split}
\end{equation*}
Using the same approach as before, considering that $\dot{\Delta} r=\left(D^T \otimes\I_2 \right)\dot{r}$, the differentiated values of $z$ can be expressed as
\begin{equation}\label{eq:dot_z}
\dot{z}=\mathcal{F}(r)\dot{r}
\end{equation}

The \textit{Formation Matrix} $\mathcal{F}(r)$ describes the linear mapping between the agent velocities $\dot{r}$ and the differentiated distance values $\dot{z}$, as illustrated in Fig.~\ref{fig:Mapping}.
As can be appreciated, the range space of $\mathcal{F}(r)$ is the subspace $\mathcal{R}(F) \subset \mathbb{R}^{m}$ of differentiated distance values $\dot{z}$ that can be generated by the agent velocities $\dot{r}$. The null space of $\mathcal{F}(r)$ is the subspace $\mathcal{N}(F) \subset \mathbb{R}^{2n}$ of agent velocities that do not generate changes in the distances $z$.

A distance-based control strategy exploiting the null space $\mathcal{N}(F)$ of the \textit{Formation Matrix} $\mathcal{F}(r)$ was presented by the authors in \cite{crespo2025distance}.

\begin{figure}[H]
\begin{centering}
\includegraphics[width=0.8\columnwidth]{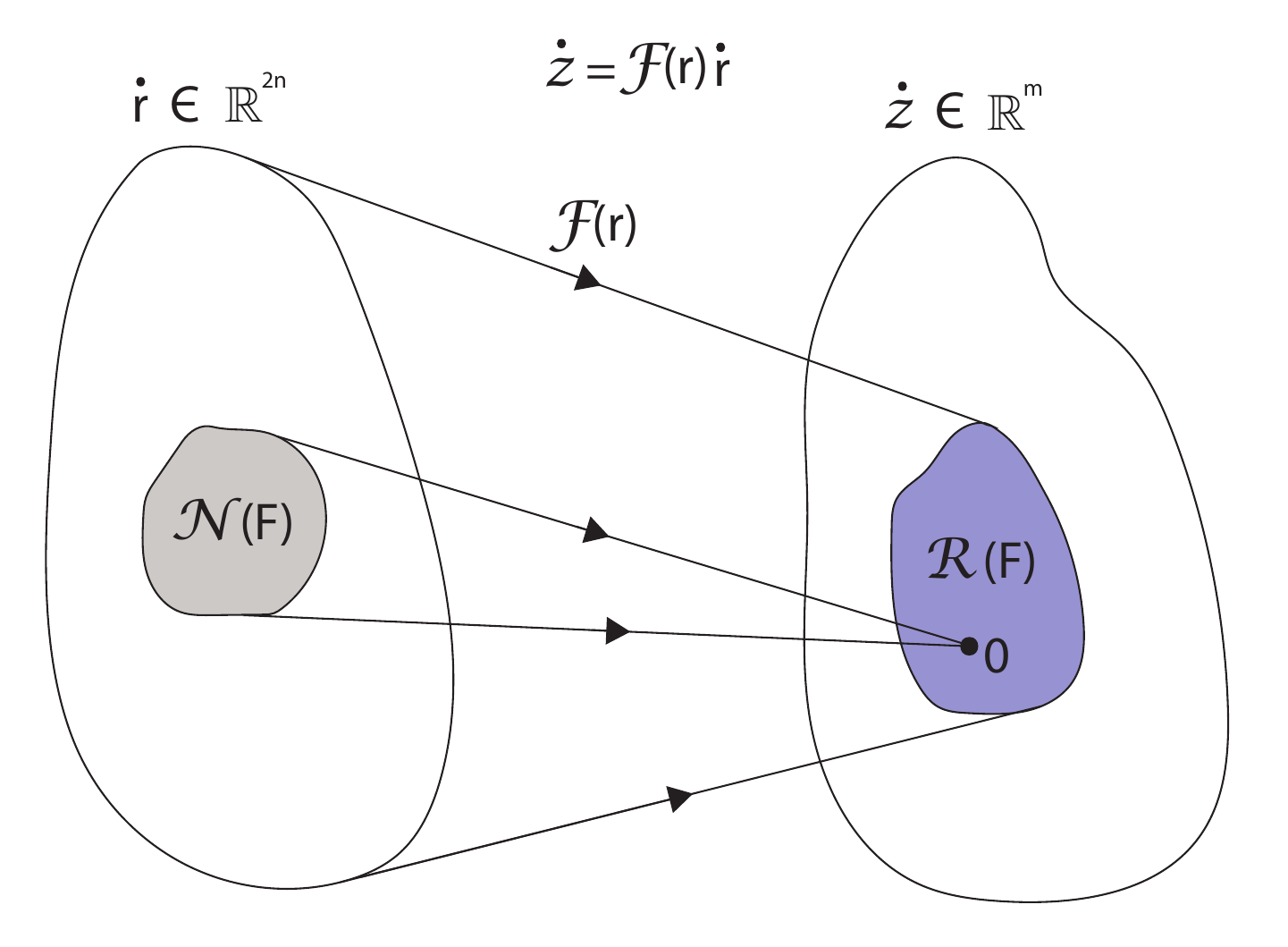}
\par\end{centering}
\caption{\label{fig:Mapping}Mapping between the agent velocities $\dot{r}$ and the differentiated distance values $\dot{z}$.}
\end{figure}

\begin{lemma}\label{prop:partial_Fq}
The following equality holds
\begin{equation*}
\left(\dfrac{\partial \left( \mathcal{F}(r)r\right)}{\partial  r}\right)^T=\mathcal{F}(r)
\end{equation*}
\end{lemma}
\begin{proof}
Differentiating equation (\ref{eq:distance_2}) with respect to time and applying the chain rule, we obtain
\begin{equation*}
\begin{split}
\dot{z}&=\dfrac{d}{dt}z \\
&=\dfrac{d}{dt}\left(\mathcal{F}(r)r\right)\\
&=\left(\dfrac{\partial \left( \mathcal{F}(r)r\right)}{\partial  r}\right)^T \dot{r}
\end{split}
\end{equation*}
By comparing this with equation (\ref{eq:dot_z}) we derive Lemma \ref{prop:partial_Fq}.
\end{proof}

\begin{proposition}\label{prop:holonomic}
Defining the vector $\psi:=[r^T,\;z^T]^T\in\mathbb{R}^{2n+m}$, the $m$ kinematic constraints of equations (\ref{eq:dot_z}) can be expressed in \textit{Pfaffian} kinematic matrix form
\begin{equation}\label{eq:restric_cinema_F}
A(\psi)^T\dot{\psi}=\0_{m\times 1}
\end{equation}
with 
\begin{equation*}
A(\psi)=\begin{bmatrix}
       \mathcal{F}(r)^T\\[0.3em]
		-\I_m
\end{bmatrix}
\end{equation*}
which are holonomic (or \textit{integrable}).
\end{proposition}
\begin{proof}
The existence of the $m$ holonomic constraints (\ref{eq:distance_2}) can be expressed in the form
\begin{equation}\label{eq:restric_holo_F}
\begin{split}
h(\psi)&=\mathcal{F}(r)r-z\\
&=\0_{m\times 1}
\end{split}
\end{equation}
and implies an equal number of kinematic constraints:
\begin{equation}\label{eq:restric_cine_h}
\begin{split}
\dfrac{d h_i(\psi)}{dt}&=\dfrac{\partial h_i(\psi)}{\partial \psi}\dot{\psi}=0 \qquad i=1,\cdots,m
\end{split}
\end{equation}
In this case, considering Lemma \ref{prop:partial_Fq}, it is satisfied that
\begin{equation*}
\begin{split}
\dfrac{d}{dt}h(\psi)&=\left(\dfrac{\partial h(\psi)}{\partial \psi}\right)^T\dot{\psi}\\
&=\left[\left(\dfrac{\partial \left( \mathcal{F}(r)r\right)}{\partial  r}\right)^T\quad -\left(\dfrac{\partial z}{\partial z}\right)^T \right]\dot{\psi}\\
&=\left[\mathcal{F}(r)\quad -\I_m \right]\dot{\psi}\\
&=A(\psi)^T\dot{\psi} \quad \Rightarrow \quad \int A(\psi)^T\dot{\psi}=h(\psi)
\end{split}
\end{equation*}
Since the $m$ kinematic constraints in equation (\ref{eq:restric_cinema_F}) are integrable and can be expressed as $m$ constraints of the form (\ref{eq:restric_holo_F}), we conclude that they are holonomic \cite{ne_mark2004dynamics}.
\end{proof}

\begin{proposition}\label{pro:RankF}
If the graph $\mathcal{G}$ with $n$ nodes and $m$ edges is minimally rigid, then the \textit{Formation Matrix} $\mathcal{F}(r)$ has full rank equal to $m=2n-3$ for a non-degenerate configuration.
\end{proposition} 
\begin{proof}
A graph $\mathcal{G}$ with $n\geq 2$ nodes is minimally rigid if and only if it has $2n - 3$ edges and each subset of $n' \leq n$ nodes has no more than $2n' - 3$ edges \cite{mesbahi2010graph}.

For a subset of $n' = 2$ nodes, it is possible to define independently the length of only one edge.
Generally, for each subset of $n' \leq n$ nodes, it is possible to independently define the length of $2n' - 3$ edges.
If $\mathcal{G}$ is minimally rigid, each subset of $n' \leq n$ nodes has no more than $2n' - 3$ edges, so the \textit{Formation Matrix} $\mathcal{F}(r)$ has $m=2n-3$ linearly independent rows. 
Then $\mathcal{F}(r)$ has full rank equal to $m$.
\end{proof}

\section{Interconnected Physical Structure}\label{sec:phy_strucuture}

In this section, we present the pH and BG energy-based models for both a single ECB and the MAS. These models underscore the energy storage, the interconnection structure, and the dynamic behavior of the spring of the ECB.

\subsection{ECB Models}

Next, we present the pH and BG model of the ECB.

\subsubsection{pH ECB-Model}
Let's consider a new variable $\tilde{z}_k$, defined as the difference between the length $z_k$ of the $k$-th virtual spring and its de-energized value $z_k^*$, i.e.,  
\begin{equation*}
\tilde{z}_k:=z_k-z_k^*
\end{equation*}

The effort $f_{rk}$ exerted by the virtual springs of each ECB$_k$ is characterized by a barrier function that tends to infinity as the distance approaches zero.  
This results in the spring operating in compression with negative forces $f_{rk}$ tending to infinity as the distance $z_k$ approaches zero, thereby repelling agents.  
For distances $z_k$ greater than the reference $z_k^*$, the virtual spring operates in traction, exerting positive forces that attract both agents.
It can be expressed as
\begin{equation*}
\begin{split}
f_{rk}(\tilde{z}_k; z_k^*)&=\dfrac{\partial H_{ck}}{\partial \tilde{z}_k}
\end{split}
\end{equation*}
where 
\begin{equation}\label{eq:Hamitonian_ECB}
H_{ck}(\tilde{z}_k)=H_{sk}(\tilde{z}_k)+\Phi_{sk}(z_k^*)
\end{equation}
is the Hamiltonian function of the ECB$_k$ that corresponds to the total (potential) energy of the spring.
The value of the constant term $\Phi_{sk}(z_k^*)$ is chosen so that the Hamiltonian has a minimum at $z_k=z_k^*$.

In this way the Hamiltonian model of the ECB$_k$ is
\begin{equation}\label{eq:model_ECBk}
\begin{cases}
\dot{\tilde{z}}_k&=\zeta_k\\
f_k&=\dfrac{\partial H_{ck}}{\partial  \tilde{z}_k}+b_k\zeta_k
\end{cases}
\end{equation}
where the parameter $b_k$ is a linear damping coefficient, the input $\zeta_k$ the change rate of the length of the virtual spring, and the output $f_k$ the total effort exerted by the ECB$_k$.

\begin{remark}
Since the length of the de-energized spring can be interpreted as the desired value $z_k^*$, is logic to consider it constant so that $\dot{z}_k^*=0$. 
\end{remark}

\subsubsection{BG ECB-Model}

The BG model consists of \textbf{R} and \textbf{C} elements, linked by a 1-junction, representing the dissipation and potential energy-storing phenomena, respectively.
The causality convention defines the velocity $\zeta_k$ as the input and the effort $f_k$ as the output.
The corresponding BG model of the ECB$_k$ is shown in Fig. \ref{fig:BG_ECB}.

\begin{figure}[H]
\begin{centering}
\includegraphics[width=0.8\columnwidth]{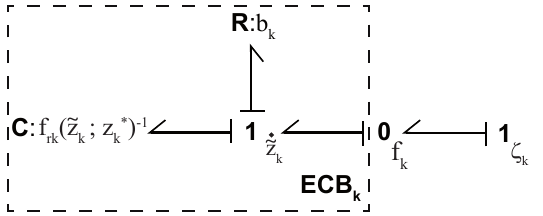}
\par\end{centering}
\caption{\label{fig:BG_ECB}BG scheme of the ECB$_k$.}
\end{figure}

\subsection{MAS Models}

In this section we present the pH and BG model of the IPS depicted in Fig \ref{fig:IF}.

\subsubsection{pH MAS-Model}

The effort $f_k$ exerted by the ECB$_{k}$, considering that the node $i$ is the positive end of edge $k$ and $j$ the negative one (see Fig. \ref{fig:Rotation}), can be decoupled in 
\begin{equation*}
\begin{split}
\begin{bmatrix}
       F_{i}\\[0.3em]
	   F_{j}\\[0.3em]
\end{bmatrix}
&=-\left(\begin{bmatrix}
       1 \\[0.3em]
       -1
\end{bmatrix}\otimes \I_2 \right)
\begin{bmatrix}
       c_k\\[0.3em]
	   s_k
\end{bmatrix}f_{k}
\end{split}
\end{equation*}
where $F_{i}={\left[F_{ix},\;F_{iy}\right]}^T$, $F_{j}={\left[F_{jx},\;F_{jy}\right]}^T$ represent the forces, with respect to a fixed reference frame, exerted by ECB$_{k}$ on agents $i$ and $j$, respectively.
Considering the incidence matrix $D$ in the same way done before, the efforts exerted on the agents by all the ECBs can be expressed as
\begin{equation}\label{eq:fuerzas_ECB}
F=-\mathcal{F}(r)^T f
\end{equation}
where $f=[f_1,\cdots, f_m]^T\in\mathbb{R}^m$ and $F=\left[F_1^T,\cdots,F_n^T\right]^T\in\mathbb{R}^{2n}$.

Thus, from equation (\ref{eq:pto_mat_din}), we have that the resulting momentum of the agents, with respect to a fixed reference frame, is
\begin{equation}\label{eq:pdotexpression}
\dot{p}=F + \nu
\end{equation}
where $\nu=[\nu_1^T,\cdots, \nu_n^T]^T\in\mathbb{R}^{2n}$ is a new external signal reserved to eventually add another control law.

Then, from equation (\ref{eq:dot_z}) and (\ref{eq:pdotexpression}), the interconnection that represents the interaction between the $m$ ECBs and the $n$ agents can be mathematically expressed in a pH form in terms of the shared velocities and efforts as
\begin{equation}\label{eq:sigma_I}
\Sigma_{I}:\begin{cases}
\begin{bmatrix}
       \dot{p}\\[0.3em]
	   \zeta\\[0.3em]
\end{bmatrix}=
\begin{bmatrix}
       \0_{2n}& -\mathcal{F}(r)^T\\[0.3em]
	   \mathcal{F}(r) & \0_m \\[0.3em]
\end{bmatrix}
\begin{bmatrix}
       \dot{r}\\[0.3em]
	   f \\[0.3em]
\end{bmatrix}+
\begin{bmatrix}
       \nu\\[0.3em]
	   \0_{m\times 1} \\[0.3em]
\end{bmatrix}
\end{cases}
\end{equation}
where $\zeta=[\zeta_1,\cdots, \zeta_m]^T\in\mathbb{R}^m$. 
Note that since $\dot{z}^*=0$, then $\zeta=\dot{z}$ and it relates to the velocities $\dot{r}$ through the \textit{Formation Matrix} as expressed in (\ref{eq:dot_z}).

Considering $\nu=0$ for the moment, the inputs and outputs of $\Sigma_{I}$ satisfy
\begin{equation}\label{eq:sigmaI_e}
\begin{split}
\zeta^Tf&=\dot{r}^T\mathcal{F}(r)^T f\\
&=-\dot{r}^T\dot{p}+\dot{r}^T\nu\Rightarrow \zeta^Tf+\dot{r}^T\dot{p}=0
\end{split}
\end{equation}
Thus, the interconnection subsystem is power-preserving.

Considering equation (\ref{eq:pdotexpression}) and that the Hamiltonian $H_p:\mathbb{R}^{2n}\rightarrow \mathbb{R}$ of the system (\ref{eq:pto_mat_din}) is its kinetic energy, i.e.,
\begin{equation}\label{eq:Hamiltoniano_MAS}
H_p(p)=\dfrac{1}{2}p^TM^{-1}p
\end{equation}	
the pHS model of the $n$ agents is
\begin{equation}\label{eq:sigma_P}
\Sigma_{MAS}:\begin{cases}
\begin{bmatrix}
       \dot{r}\\[0.3em]
	   \dot{p}
\end{bmatrix}&=
\begin{bmatrix}
       \0_{2n}& \I_{2n}\\[0.3em]
	   -\I_{2n} & \0_{2n} 
\end{bmatrix}
\nabla H_p
+
\begin{bmatrix}
       \0_{2n}\\[0.3em]
	   \I_{2n}
\end{bmatrix}\left( F + \nu \right)\\
\dot{r}&=\begin{bmatrix}
       \0_{2n} & \I_{2n}
\end{bmatrix}
\nabla H_p
\end{cases}
\end{equation}
Moreover, since
\begin{equation}\label{eq:Hp_e}
\dot{H}_p=\dot{r}^T\dot{p}
\end{equation}
system $\Sigma_{MAS}$ is passive (lossless) \cite{ortega2001putting}.

Finally, from the model in (\ref{eq:model_ECBk}), the dynamics of the set of ECBs can be expressed as
\begin{equation*}
\Sigma_{ECB}:\begin{cases}
\dot{\tilde{z}}&=\zeta\\
f&=\dfrac{\partial H_c}{\partial  \tilde{z}}+B\zeta
\end{cases}
\end{equation*}
where $\tilde{z}={\left[\tilde{z}_1,\cdots,\tilde{z}_m\right]}^T\in\mathbb{R}^{m}$, $B=\text{diag}\{[b_1,\cdots, b_m]\}$ and $H_c:\mathbb{R}^{m}\rightarrow \mathbb{R}$ the Hamiltonian with the expression
\begin{equation}\label{eq:Hamiltoniano_C}
\begin{split}
H_c(\tilde{z})=&H_s(\tilde{z})+\Phi_s\left(z^*\right)
\end{split}
\end{equation}
being
\begin{equation*}
\begin{split}
H_s(\tilde{z})=&\sum_{k=1}^{m}H_{sk}(\tilde{z}_k) \qquad \Phi_s\left(z^*\right)=\sum_{k=1}^{m}\Phi_{sk}\left(z_k^*\right)
\end{split}
\end{equation*}
the sum of the Hamiltonians defined in (\ref{eq:Hamitonian_ECB}).
Considering that $B\geq 0$, then
\begin{equation}\label{eq:Hc_e}
\begin{split}
\dot{H}_c&=\zeta^T \dfrac{\partial H_c}{\partial \tilde{z}}\\
&\leq \zeta^Tf\\
\end{split}
\end{equation}
and the system $\Sigma_{ECB}$ is passive.

Defining a new energy function $W$ as the sum of the Hamiltonians of the agents (\ref{eq:Hamiltoniano_MAS}) and the ECBs (\ref{eq:Hamiltoniano_C}),
\begin{equation}\label{eq:W_CbI}
\begin{split}
W(p,\tilde{z}):=&H_p(p)+H_s(\tilde{z})+\Phi_s\left(z^*\right)
\end{split}
\end{equation}
the complete dynamics of the IPS can also be expressed in a pH form in terms of $W$ as
\begin{equation}\label{eq:dynamics_CbI_1}
\begin{bmatrix}
       \dot{r}\\
       \dot{p}\\
	   \dot{\tilde{z}}
\end{bmatrix}=
\begin{bmatrix}
		\0  & \I & \0\\
       -\I &-\mathcal{F}(r)^TB\mathcal{F}(r)& -\mathcal{F}(r)^T\\
	    \0 & \mathcal{F}(r) & \0
\end{bmatrix}\nabla W
+
\begin{bmatrix}
	   \0 \\
       \nu\\
	   \0 \\
\end{bmatrix}
\end{equation}

\begin{remark}
The use of virtual springs can be seen, from an energy-related perspective, as a way to shape the energy (\ref{eq:W_CbI}) through the selection of $H_s$ and $\Phi_s$.
\end{remark}

\subsubsection{BG MAS-Model}

The BG model of the MAS is depicted in Fig. \ref{fig:BG_IF_2}, where three different parts can be distinguished:
(i) the collection of all the ECB models grouped under $\Sigma_{ECB}$,
(ii) the block $\Sigma_{MAS}$ containing $n$ storage elements \textbf{I} representing the agents modeled as point masses%
, and
(iii) the interconnection block $\Sigma_I$ representing the coupling structure.
The later one consist of a power-conserving transformation element \textbf{MTF} modulated by the transpose of the \textit{Formation Matrix}, that relates the efforts $f$ with $F$ and the flows $\dot{r}$ with  $\zeta$.

A major advantage of modeling with BG is that the expression of the control law can be straightforwardly obtained by reading the input-output dependence on the power bond at the interface/interconnection of the $\Sigma_{MAS}$ and the IPS (with $\dot{r}$ as input and $\dot{p}$ as output, as dictated by the causal assignment).

\begin{figure}[H]
\centering
\includegraphics[width=\columnwidth]{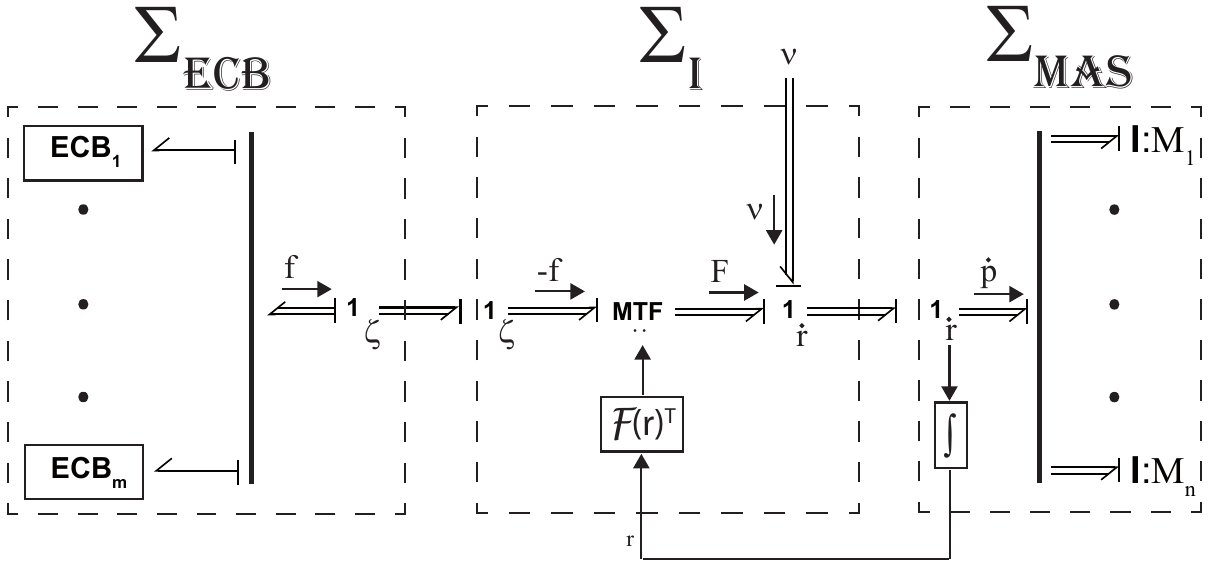}
\caption{\label{fig:BG_IF_2}BG scheme of the IPS.}
\end{figure}

\section{Formation Control}\label{sec:FC}

In this section we demonstrate that the BG model in Fig. \ref{fig:BG_IF_2}, can be transformed into a CbI scheme, with the agents and the ECBs modeled as pHS.
In doing so we push forward the assignment of physical attributes to the graph: as we associate to each edge an ECB, the overall interconnection structure of the MAS can be related to the \textit{Formation Matrix}.
Furthermore, we reveal that the CbI scheme naturally generates \textit{Casimir} functions and that all the mathematical properties associated with the proposed IPS emerge naturally.
This results then in a static feedback control scheme.
Finally, we extend the analysis to bases modeled as rigid bodies.

\subsection{Control by Interconnection}\label{sec:cbi}

Let's consider that the agents form the plant and the collection of ECBs constitutes the controller, then the BG variables are replaced with those associated with the CbI theory, according to the usual topology as presented in \cite{ortega2004interconnection}, for instance. The conversion is shown in Table \ref{tab:model_variables} and the corresponding CbI scheme is presented in Fig. \ref{fig:CbI}.

\begin{table}[H]
\begin{center}
\caption{Table of variables and their expressions.}
\label{tab:model_variables}
\begin{tabular}{| c | c | c | c | c |}
\hline\rule{0pt}{3ex}
BG & CbI & Type & Physical Quantity & Expression \\ \hline \rule{0pt}{3ex}
$\dot{p}$ & $u$ & Input & Force & $-\mathcal{F}(r)^T f+\nu$  \\ \hline \rule{0pt}{3ex}
$\dot{r}$ & $y$ & Output & Velocity & $M^{-1}p$  \\ \hline\rule{0pt}{3ex}
$\zeta$ & $u_c$ & Input & Velocity & $\mathcal{F}(r)\dot{r}$  \\ \hline\rule{0pt}{4ex}
$f$ & $y_c$ & Output & Force & $\dfrac{\partial H_c}{\partial\tilde{z}}+B\zeta$  \\ \hline
\end{tabular}
\end{center}
\end{table}

\begin{figure}[b!]
\centering
\includegraphics[width=\columnwidth]{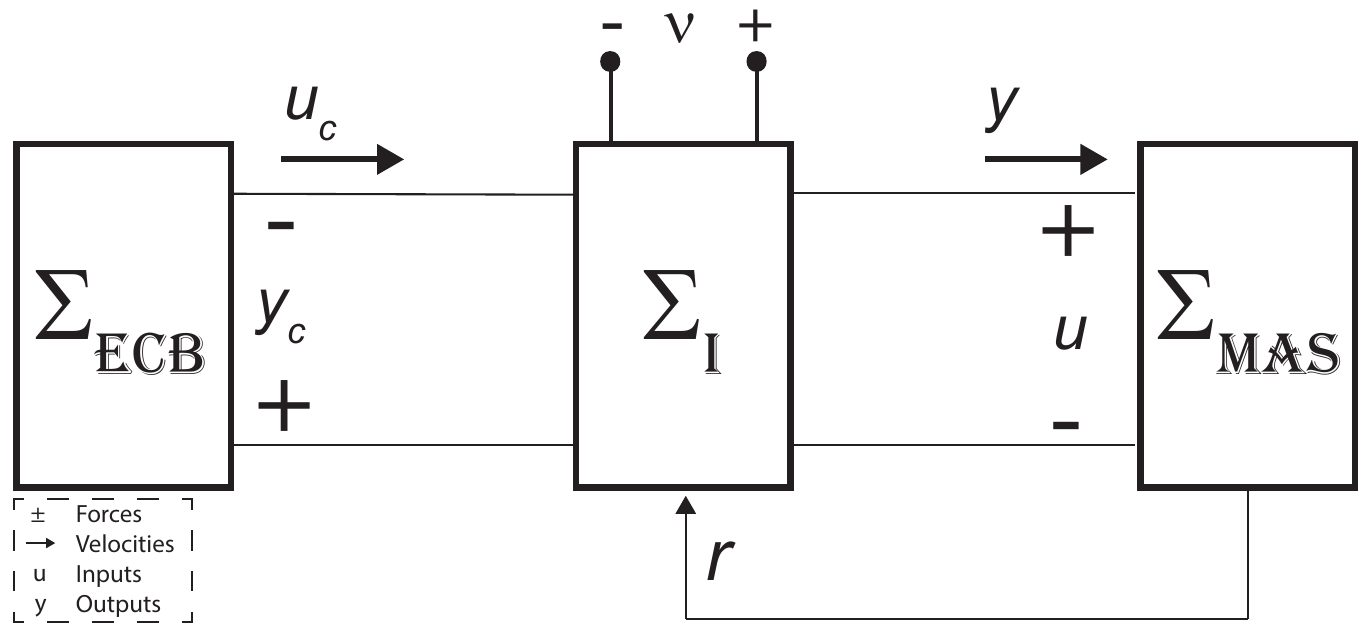}
\caption{\label{fig:CbI}Block diagram of the CbI scheme.}
\end{figure}

In the CbI method, we look for functions that are dynamically invariant and independent of the Hamiltonian. That is, to assign a desired form to the global energy function, it is necessary to relate the states of the plant $r$ and the controller $\tilde{z}$ by generating invariant sets defined as \textit{Casimir} functions \cite{ortega2008control}.

\begin{proposition}\label{prop:casimir}
The function 
\begin{equation*}
\mathcal{K}(r,\tilde{z})=\mathcal{F}(r)r-\tilde{z}
\end{equation*}
qualifies as a \textit{Casimir} function, relating the states of the plant $r$ and the controller $\tilde{z}$ through invariant sets.
\end{proposition}
\begin{proof}
Let define the set
\begin{equation}\label{eq:conjunto_omega_c}
\Omega_c:=\{(r,\tilde{z})\in \mathbb{R}^{2n} \times \mathbb{R}^{m}|\tilde{z}=\mathcal{C}(r)-c\}\, ,\; c\in\mathbb{R}^{m}
\end{equation}
with $\mathcal{C}(r)=\mathcal{F}(r)r$ and $c=z^*$.
A necessary and sufficient condition for all sets $\Omega_c$, known as the foliation of the manifold $\Omega_C$, to be invariant%
\footnote{The set $\Omega_c \subset \mathbb{R}^{2n} \times \mathbb{R}^{m}$ is invariant if $\left(r(0),\tilde{z}(0)\right) \in \Omega_c \;\Rightarrow\; \left(r(t),\tilde{z}(t)\right)\in \Omega_c$ for all $t \,\geq \,0$.}
is that $\dot{\mathcal{K}}=0$.
Since the time derivative $\dot{\mathcal{K}}$ is zero along the closed-loop dynamics (\ref{eq:dynamics_CbI_1}) for all Hamiltonians $H_p$, that is
\begin{equation*}
\dot{\mathcal{C}}(r)-\dot{\tilde{z}}=0
\end{equation*}
This, and considering that the function $\mathcal{K}(r,\tilde{z})$ is independent of the energy functions $H_p$ and $H_c$ completes the proof.
\end{proof}

\begin{remark}
In this case, the constant $c$ of the set (\ref{eq:conjunto_omega_c}) has a physical meaning: choosing $c=z^*$ defines the length of the de-energized spring.
\end{remark}

\begin{remark}
Projecting the system's evolution onto $\Omega_c$, the dynamics (\ref{eq:dynamics_CbI_1}) can be expressed in a reduced form as
\begin{equation*}
\begin{bmatrix}
       \dot{r}\\[0.3em]
       \dot{p}\\
\end{bmatrix}=
\begin{bmatrix}
		\0_{2n}  & \I_{2n} \\[0.3em]
       -\I_{2n} &-\mathcal{F}(r)^TB\mathcal{F}(r)\\
\end{bmatrix}\nabla W_c
+
\begin{bmatrix}
	   \0_{2n\times 1} \\
       \nu
\end{bmatrix}
\end{equation*}
where the energy (\ref{eq:W_CbI}) has been replaced by
\begin{equation*}
W_c(r,p):=H_p(p)+H_s(\mathcal{F}(r)r-z^*)+\Phi_s\left(z^*\right)
\end{equation*}
\end{remark}

\begin{remark}\label{rem:initial_cond}
As the level sets of the \textit{Casimir} function are invariant, to ensure that the trajectory begins and remains in $\Omega_c$, they must satisfy
\begin{equation*}
\mathcal{F}\left(r(0)\right)r(0)-\tilde{z}(0)=\mathcal{F}(r(t))r(t)-\tilde{z}(t)
\end{equation*}
for all $t \,\geq \,0$. Since $\mathcal{F}\left(r(t)\right) r(t)=z(t)$, the initial conditions of the controller's states $z$ must satisfy
\begin{equation*}
\begin{split}
z(0)&=\mathcal{F}(r(0))r(0)\\
\end{split}
\end{equation*}
for any vector $r(0)$. In other words, it is not possible to assign arbitrary values to the initial values $z(0)$ given the initial positions $r(0)$. Beyond that, imposing a restriction on the initial condition of the controller is rather unnatural, so we will address this issue in the next section.
\end{remark}

\begin{remark}
In this work the deduction of the function $\mathcal{C}(r)=\mathcal{F}(r)r$ of the set (\ref{eq:conjunto_omega_c}) arises naturally from the approach used. This is a significant advantage, as in general, obtaining such functions involves 
solving partial differential equations (PDEs) of the form
\begin{equation*}
\begin{bmatrix}
	   (\nabla \mathcal{C})^T & -\I_m
\end{bmatrix}
\begin{bmatrix}
       \dot{r}\\
	   \dot{\tilde{z}}
\end{bmatrix}
=\0_{m\times 1}
\end{equation*}
where in this case, as pointed in Lemma \ref{prop:partial_Fq} the solution holds
\begin{equation*}
\begin{split}
(\nabla \mathcal{C})^T&=\left(\dfrac{\partial \left( \mathcal{F}(r)r\right)}{\partial  r}\right)^T\\
&=\mathcal{F}(r)
\end{split}
\end{equation*}
\end{remark}

\subsection{State Feedback Formation Control Law}

Proposition \ref{prop:holonomic} demonstrates that $z$ can be obtained directly from the \textit{Formation Matrix}.
As a result, there is no need to integrate the velocities $\zeta$, and the disadvantage related to the initial condition, highlighted in Remark \ref{rem:initial_cond}, is avoided.
This implies that the proposed IPS in Fig. \ref{fig:IF} eliminates, within a feedback scheme representation, the dynamic behavior associated with the virtual springs. 

Thus the CbI approach provides a static state-feedback solution to the original formation control problem,  with  the control law
\begin{equation}\label{eq:feedback_control_law}
u=-\mathcal{F}(r)^T \left( f_r(\tilde{z}; z^*)+B\mathcal{F}(r)\dot{r} \right) +\nu\\ 
\end{equation}
and the corresponding control block diagram of Fig. \ref{fig:DB_control_rigid}.

\begin{figure}[b!]
\centering
\includegraphics[width=\columnwidth]{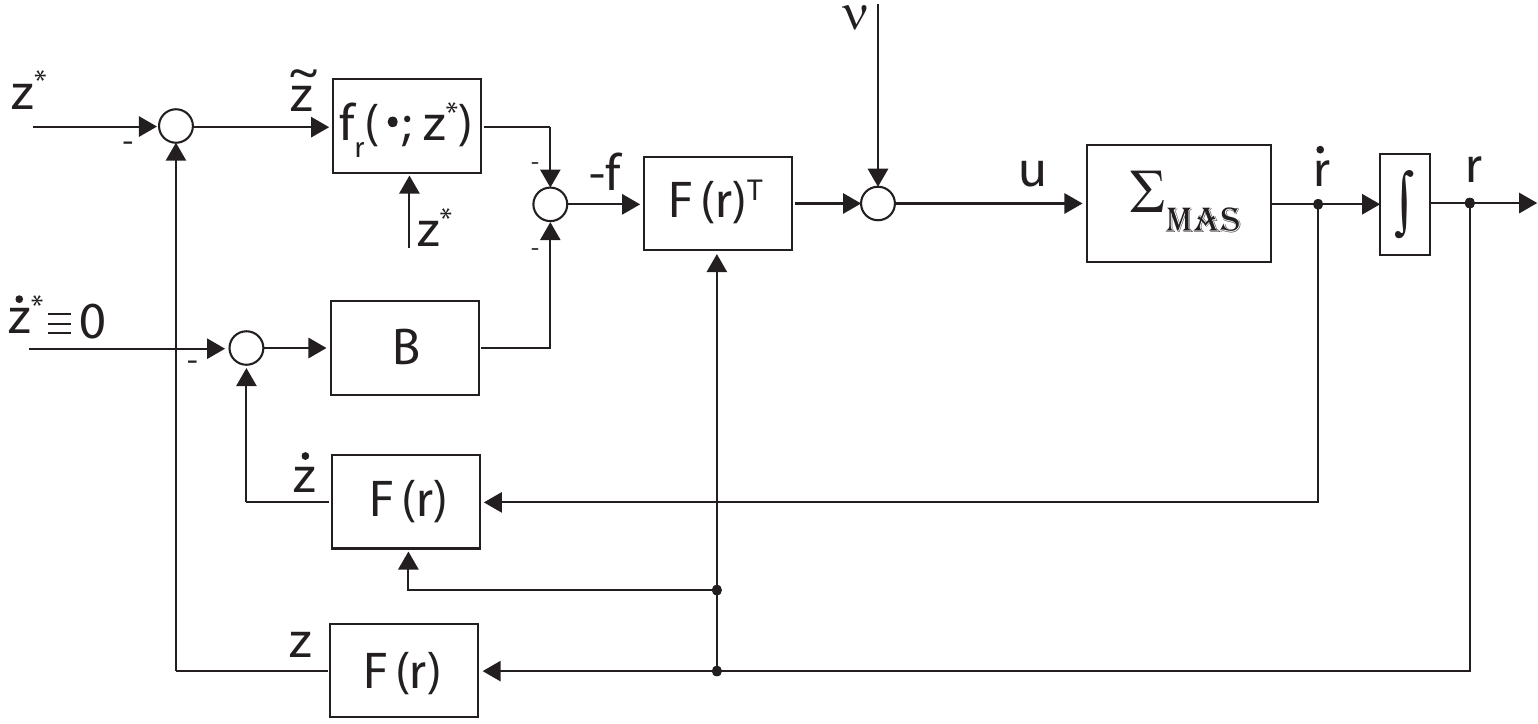}
\caption{\label{fig:DB_control_rigid} Multi-agent control scheme diagram.}
\end{figure}

\subsection{Extension to Rigid Body Robots Models}

In this section, we assume that the agents are omnidirectional mobile bases modeled as in (\ref{eq:MB_dynamic_equation}).
We split the control input into two components:  one to control the positions $r$ of the centers of mass and the other to control the angles $\varphi$. 

In that sense, we introduce the transformation matrices $T_{r_i}$ and $T_{\varphi_i}$ to split the coordinates $q_i$ into $r_i$ and $\varphi_i$, respectively, as
\begin{equation*}
\begin{bmatrix}
       r_i \\[0.3em]
	  \varphi_i
\end{bmatrix}=
\begin{bmatrix}
       T_{r_i} \\[0.3em]
	   T_{{\varphi}_i}
\end{bmatrix}
q_i; \quad
T_{r_i}=\begin{bmatrix}
       1 & 0 & 0 \\[0.3em]
	   0 & 1 & 0 \\[0.3em]
\end{bmatrix}\quad
T_{{\varphi}_i}=\begin{bmatrix}
       0 & 0 & 1 \\
\end{bmatrix}
\end{equation*}
for $i=1,\cdots,n$.
Considering vector $q$, the transformation can be expressed in a compact form as:
\begin{equation*}
\begin{bmatrix}
       r \\[0.3em]
	  \varphi
\end{bmatrix}=
\begin{bmatrix}
       T_r \\[0.3em]
	   T_{\varphi}
\end{bmatrix}
q
\end{equation*}
where
$\varphi=\left[\varphi_1,\cdots,\varphi_n\right]^T\in\mathbb{R}^{n}$,
and
\begin{equation*}
T_r=\begin{bmatrix}
       T_{r_1} & \0_{2\times 3} & \cdots & \0_{2\times 3}\\
       \0_{2\times 3} & \ddots &  &\vdots \\
	   \vdots &  & \ddots & \0_{2\times 3} \\
	   \0_{2\times 3} & \cdots & \0_{2\times 3} & T_{r_n} \\
\end{bmatrix}
\end{equation*}

\begin{equation*}
T_{\varphi}=\begin{bmatrix}
       T_{\varphi_1} & \0_{1\times 3} & \cdots & \0_{1\times 3}\\
       \0_{1\times 3} & \ddots &  &\vdots \\
	   \vdots &  & \ddots & \0_{1\times 3} \\
	   \0_{1\times 3} & \cdots & \0_{1\times 3} & T_{\varphi_n} \\
\end{bmatrix}
\end{equation*}

As the mobile bases are omnidirectional, the components of the coordinates $q_i$ are independent of each other.
Consequently, the base translation can be controlled with the same law derived for the point mass model, while the rotation of the base can be controlled by a separate law, as shown next.

Let's choose in (\ref{eq:MB_dynamic_equation}) the inverse dynamic control law $\tau = u_M$
\begin{equation*}
u_M=G(q)^{-1}\left(u_r+u_{\varphi}+C(q,\dot{q})\dot{q}\right)
\end{equation*}
with\footnote{For the sake of simplicity in the specification of the control law, we consider all the bases to be identical.}
\begin{equation*}
\begin{split}
u_r&=T_r^T\left(-\mathcal{F}(r)^T f+\nu\right)\\
u_{\varphi}&=T_{\varphi}^T\left(\frac{3J_{r}L^2}{2R^2}+I_{b}\right)\left(\ddot{\varphi^*}+K_{D_{\varphi}}\dot{\tilde{\varphi}}+K_{P_{\varphi}}\tilde{\varphi}\right)\end{split}
\end{equation*}
where $K_{D_{\varphi}}$ and $K_{P_{\varphi}}$ are positive definite (diagonal) matrices,
$\varphi^*$ and $\tilde{\varphi}=\varphi^*-\varphi$ are the angle position references and errors, respectively.
Then, we obtain two separate dynamics, one for $r$ and another for $\varphi$ as:
\begin{equation*}
\begin{split}
&\left(\frac{3J_{r}}{2R^2}+m_{b}\right)\ddot{r}=-\mathcal{F}(r)^T f+\nu\\
&\ddot{\tilde{\varphi}}+K_D\dot{\tilde{\varphi}}+K_P\tilde{\varphi}=0
\end{split}
\end{equation*}
Both control laws are directly translated into torque commands for the wheel motors through the inverse of the matrix $G(q)$, with the latter being managed by low-level motor controllers.

\begin{remark}
The dynamics for $r$ are the same as those for the point mass model in equation (\ref{eq:pdotexpression}). 
Thus, the control law design and stability analysis for the omnidirectional bases can be performed using the point mass model.
\end{remark}

\section{Leader Agent Control}\label{sec:LAC}
In this section we enhance the formation control law with a leader-follower strategy to place the interconnected MAS at a specific position and orientation in the plane. Finally, we demonstrate convergence to the equilibrium point through the direct method of Lyapunov.

To this end, we select any two neighboring nodes, define one of them as the leader, and apply an additional control law to regulate the position and orientation of the MAS.

Let's consider the leader node and one of the edges that connects it to one of its neighboring agents. 
As the IPS of Fig. \ref{fig:IF} behaves like a semi-rigid body, by fixing the position of the leader and the orientation of the edge, the three degrees of freedom for the entire MAS in the plane are determined.

Given the velocities $\dot{x}_i$, $\dot{y}_i$ and $\dot{x}_j$, $\dot{y}_j$ of  two neighboring agents, as shown in Fig. \ref{fig:Rotation}, the velocities perpendicular to the edge $k$ that connects both agents can be expressed as
\begin{equation}\label{eq:ve_tan}
\begin{split}
v_i&=-\dot{x}_is_k+\dot{y}_ic_k \quad v_j=\dot{x}_js_k-\dot{y}_jc_k
\end{split}
\end{equation}
Considering that the rotation velocity $\omega_k=\dot{\alpha}_k$ can be expressed as
\begin{equation*}
\omega_k=\dfrac{v_i+v_j}{z_k}
\end{equation*}
replacing by (\ref{eq:ve_tan}) it can be rewritten as 
\begin{equation*}
\omega_k=-\dfrac{s_k}{z_k}\dot{x}_i+\dfrac{c_k}{z_k}\dot{y}_i+\dfrac{s_k}{z_k}\dot{x}_j-\dfrac{c_k}{z_k}\dot{y}_j
\end{equation*}

Consider now that agent $i$ is the leader. Let us define the vector $a:=\left[x_i,\;y_i,\;\alpha_k\right]^T$, which represents its position and the orientation angle of edge $k$. 
Then, the time derivative of the vector $a$ can be expressed as:
\begin{equation}\label{eq:dot_a}
\dot{a}=\mathcal{A}(r)\dot{r}
\end{equation}
where $\mathcal{A}(r)$, of dimension $({3\times 2n})$, is defined as the \textit{Leader} matrix, and its expression is
\begin{equation*}
\begin{split}
\mathcal{A}(r)=&
\begin{bmatrix}
	   0 & \hdots & 1 & 0 & \hdots & 0 & 0 & \hdots & 0\\[0.3em]
	   0 & \hdots & 0 & 1 & \hdots & 0 & 0 & \hdots & 0\\[0.3em]
	   0 & \hdots &-\dfrac{s_k}{z_k} & \dfrac{c_k}{z_k} & \hdots & \dfrac{s_k}{z_k} & -\dfrac{c_k}{z_k} & \hdots & 0\\[0.3em]
\end{bmatrix}\\
&_\text{Column index} \quad \uparrow_{2i-1} \; \uparrow_{2i} \qquad \quad \uparrow_{2j-1} \;\;\uparrow_{2j} \\[0.3em]
\end{split}
\end{equation*}

Thus, accounting from (\ref{eq:dot_z}) and (\ref{eq:dot_a}) the derivatives $\dot{z}$ and $\dot{a}$ can be linked to the velocities of the agents $\dot{r}$ through the transformation
\begin{equation*}
\begin{bmatrix}
       \dot{z}\\
	   \dot{a}
\end{bmatrix}=
\mathcal{T}(r)\dot{r}
\end{equation*}
where $\mathcal{T}(r):=\left[\mathcal{F}(r)^T,\;\mathcal{A}(r)^T\right]^T$ is a ($2n\times 2n$) square matrix.

\begin{proposition}\label{pro:RankT}
If the MAS is characterized by a minimally rigid graph $\mathcal{G}$, then the matrix $\mathcal{T}(r)$ has full rank equal to $2n$.
\end{proposition} 
\begin{proof}
First, according to Proposition \ref{pro:RankF} the \textit{Formation Matrix} $\mathcal{F}(r)$ has full rank equal to $m$. 
Second, the \textit{Leader} matrix $\mathcal{A}(r)$ has full rank equal to $3$, as all its rows are linearly independent.
Furthermore, since $m=2n-3$ due to the minimally rigid graph $\mathcal{G}$ and the 3 rows of $\mathcal{A}(r)$ are independent of the rows of the matrix $\mathcal{F}(r)$, the matrix $\mathcal{T}(r)$ has full rank equal to $m+3=2n$.
\end{proof}

\begin{proposition}
Consider the dynamics model of MAS (\ref{eq:dynamics_CbI_1}) characterized by a minimally rigid graph $\mathcal{G}$, in closed loop with the linear proportional-integral (PI) input
\begin{equation}\label{eq:v_leaderagent}
\nu=\mathcal{A}(r)^T\left(K_{Pa}\,\dot{\tilde{a}}+K_{Ia}\,\tilde{a}\right) 
\end{equation}
where $K_{Pa}$ and $K_{Ia}$ are positive definite matrices of dimension $({3\times 3})$ and $\tilde{a}=a^*-a$ is the position error. Then:
\begin{itemize}
\item[(i)]
The closed-loop dynamics can be written in pH form as follows:
\begin{equation}\label{eq:closed_loop_dynamics_v}
\begin{bmatrix}
       \dot{\tilde{a}}\\
       \dot{p}\\
	   \dot{\tilde{z}}
\end{bmatrix}=
\begin{bmatrix}
		\0_3  & -\mathcal{A} & \0_{3\times m} \\
       \mathcal{A} &-\mathcal{T}(r)^T\gamma\mathcal{T}(r)& -\mathcal{F}(r)^T\\
	    \0_{m\times 3} & \mathcal{F}(r) & \0_m
\end{bmatrix}\nabla V
\end{equation}
with $\mathcal{\gamma}=\text{diag}\{[B,\,K_{Pa}]\}$ a positive definite matrix of dimension $(2n\times 2n)$ and
\begin{equation}\label{eq:lyap_leader}
V(p,\tilde{z},\tilde{a})=W(p,\tilde{z})+\dfrac{1}{2}\tilde{a}^TK_{Ia}\,\tilde{a} >0, \;\forall \; p,\tilde{z} ,\tilde{a} \neq 0
\end{equation}

\item[(ii)]
The equilibrium  point $(p_*=0,\tilde{z}_*=0,\tilde{a}_*=0)$ is stable.

\item[(iii)]
If the output
\begin{equation*}
y_D=\mathcal{T}(r)M^{-1}p
\end{equation*}
is a detectable output of the dynamics (\ref{eq:closed_loop_dynamics_v}), then $(p_*=0,\tilde{z}_*=0,\tilde{a}_*=0)$ is asymptotically stable.
\end{itemize}

\end{proposition} 
\begin{proof}
First, to prove (i), we replace (\ref{eq:v_leaderagent}) in the dynamics (\ref{eq:dynamics_CbI_1}), and consider that $a^*$ is a constant reference vector.  After some simple calculations we get (\ref{eq:closed_loop_dynamics_v}).
To prove (ii), we consider (\ref{eq:lyap_leader}) as a Lyapunov candidate function for the closed loop system. Its time derivative is as follows:
\begin{equation*}
\begin{split}
\dot{V}&=\dot{r}^T\dot{p}+\zeta^Tf -\zeta^TB\zeta+\tilde{a}^TK_{Ia}\dot{\tilde{a}}\\
&=\dot{r}^T\nu -\dot{r}^T\mathcal{F}(r)^TB\mathcal{F}(r)\dot{r}-\tilde{a}^TK_{Ia}\dot{a}\\
&=-\dot{r}^T\mathcal{F}(r)^TB\mathcal{F}(r)\dot{r}-\dot{r}^T\mathcal{A}(r)^TK_{Pa}\mathcal{A}(r) \dot{r}\\
&=-\dot{r}^T\mathcal{T}(r)^T\gamma\mathcal{T}(r)\dot{r}\leq 0,
\end{split}
\end{equation*}
from which, by Proposition \ref{pro:RankT} and considering that $\dot{r}=M^{-1}p$, we conclude Lyapunov stability of the equilibrium point.
Finally, to prove (iii) notice that $y_D^*=0$ iff $y_D=0$ therefore the standard detectability assumption can be imposed to $y_D$, ensuring asymptotic stability of the equilibrium point \cite{van2016l2}.
\end{proof}

\section{Position Control}\label{sec:poscontrol}
In this section, we keep the same physical approach with ECBs presented to solve the collision avoidance problem and configure the external control signal $\nu$ from equation (\ref{eq:pdotexpression}) to implement a position-based formation control. 
Additionally, we introduce the third extra control signal $\omega$ to ensure the asymptotic stability of the equilibrium point.

We start by defining a Virtual Reference Structure (VRS), composed of a collection of $n$ Virtual Reference Points (VRP) that move in the plane along a predetermined trajectory.
In turn, each agent is defined as a Real Equivalent Agent (REA).
Then, the objective of formation control is to associate each REA with its corresponding VRP, so that the former take on the topology of the VRS.

Let's consider that the VRS  has a reference frame $O_S-X_SY_SZ_S$ fixed to the structure. 
Let $r_S=\left[x_S,\;\; y_S\right]^T$
be the vector describing the position of frame $S$ in the plane,
$r_i^*=\left[x_{i}^*,\;\; y_{i}^*\right]^T$ be the vector of coordinates of the $i$-th VRP with respect to a fixed reference frame, and
$r_i^S$ be the vector of coordinates of the $i$-th VRP expressed in frame $S$, the relationship between them can be written as:
\begin{equation*}
r_i^*=r_S+R_S \;r_i^S
\end{equation*}
where
\begin{equation*}
R_S=\begin{bmatrix}
       c_{\varphi_S} & -s_{\varphi_S}\\[0.3em]
       s_{\varphi_S} & c_{\varphi_S}
     \end{bmatrix}\; 
\end{equation*}
where ${\varphi_S}$ is the rotation angle of the frame.

\begin{proposition}
Consider the MAS dynamics (\ref{eq:dynamics_CbI_1}) in closed loop with
\begin{equation}\label{eq:v_positioncontrol}
\nu=M\ddot{r}^*+K_{Ds}\dot{\tilde{r}}+K_{Ps}\tilde{r} +\omega
\end{equation}
where $r^*=[r_1^{*T},\cdots, r_n^{*T}]^T\in\mathbb{R}^{2n}$, $\tilde{r}=r^*-r$ is the position error, $K_{Ps}$ and $K_{Ds}$ are positive definite diagonal matrices of dimension ($2n\times 2n$) and $\omega \in \mathbb{R}^{2n}$ is a new control signal. Then:
\begin{itemize}
\item[(i)]
The closed-loop dynamics can be written in pH form as follows:
\begin{equation}\label{eq:closed_loop_dynamics_v_w}
\begin{bmatrix}
       \dot{\tilde{r}}\\
       \ddot{\tilde{r}}\\
	   \dot{\tilde{z}}
\end{bmatrix}=
J\nabla V_p
+
\begin{bmatrix}
	   \0_{1\times 2n} \\
       M^{-1}\omega\\
	   \0_{1\times m} \\
\end{bmatrix}
\end{equation}
with
\begin{equation*}
J=
\begin{bmatrix}
		\0_{2n}  & M^{-1} & \0_{2n\times m}\\
       -M^{-1} &-M^{-1}\Pi M^{-1}& M^{-1}\mathcal{F}(r)^T\\
	    \0_{m\times 2n} & \mathcal{F}(r)M^{-1} & \0_{m}
\end{bmatrix}
\end{equation*}
being $\Pi=K_{Ds}+\mathcal{F}(r)^TB\mathcal{F}(r)$ a positive definite matrix of dimension $(2n\times 2n)$ and
\begin{multline}\label{eq:Lyapunov_function_2}
V_p(\tilde{r},\dot{\tilde{r}},\tilde{z})=\dfrac{1}{2}\dot{\tilde{r}}^TM\dot{\tilde{r}}+\dfrac{1}{2}\tilde{r}^TK_{Ps}\tilde{r}+H_c(\tilde{z})>0 \\ 
\quad \forall \; \tilde{r},\dot{\tilde{r}},\tilde{z} \neq 0
\end{multline}

\item[(ii)]
With $\omega=0$, the equilibrium  point $(\tilde{r}_*=0,\dot{\tilde{r}}_*=0,\tilde{z}_*=0)$ is stable.
\end{itemize}

\end{proposition} 
\begin{proof}
First, to prove (i), we replace (\ref{eq:v_positioncontrol}) in the dynamics (\ref{eq:dynamics_CbI_1}), and as the VRP maintain a rigid geometric relationship, then $\mathcal{F}(r)\dot{r}^*=\0_{1\times m}$. 
To prove (ii), we consider (\ref{eq:Lyapunov_function_2}) as a Lyapunov candidate function for the closed loop system. Its time derivative is as follows:
\begin{equation*}
\begin{split}
\dot{V}_p&=-\dot{\tilde{r}}^T\Pi\dot{\tilde{r}}\leq 0,
\end{split}
\end{equation*}
which ensures stability of the desired equilibrium.
\end{proof}

\begin{remark}\label{rem:ECB_control_posicion}
In this case, the ECBs do not serve the purpose of positioning the agents in a specific formation but rather prevent collisions among them. For this reason, the rank of the matrix $\mathcal{F}(r)$ does not influence the stability of the control law, and it is advantageous to use as many ECBs as possible.
\end{remark}

Selecting the correct values of $z^*$, ensures that if $\tilde{r}$ and $\dot{\tilde{r}}$ tend to to zero, then $\tilde{z}=0$, and the formation control is achieve. 
In the proposition below, we design the control law $\omega$ to achieve the asymptotic stability of the equilibrium point $(\tilde{r}_*=0,\dot{\tilde{r}}_*=0)$.

\begin{proposition}
Consider the position and velocity errors of the closed loop dynamics (\ref{eq:closed_loop_dynamics_v_w}) and the new variable
\begin{equation*}
\xi:=
\begin{bmatrix}
       \tilde{r}\\[0.3em]
	   \dot{\tilde{r}}\\
\end{bmatrix}
\end{equation*}

\begin{itemize}
\item[(i)]
The reduced closed-loop dynamics can be written as follows:
\begin{equation*}\label{eq:xi_dynamics}
\dot{\xi}=H\xi+D\left(\mathcal{F}(r)^Tf-\omega\right)
\end{equation*}
where $H$ and $D$ are matrices of dimensions $(4n\times 4n)$ and $(4n\times 2n)$, respectively, with expressions given by 
\begin{equation*}
H=\begin{bmatrix}
       \0_{2n} & \I_{2n} \\
       -M^{-1}K_{Ps} & -M^{-1}K_{Ds}
\end{bmatrix}\quad
D=\begin{bmatrix}
       \0_{2n} \\[0.2em]
       M^{-1}
\end{bmatrix}
\end{equation*}

\item[(ii)]
Defining $\lambda:=2D^TQ\xi \in\mathbb{R}^{2n}$, with $Q$ a positive definite matrix of size $(4n\times 4n)$, and choosing
\begin{equation}\label{eq:omega}
\omega=\dfrac{\rho}{\norm{\lambda}}\lambda
\end{equation}
if $\rho>\norm{\mathcal{F}(r)^Tf}$, then $(\tilde{r}_*=0,\dot{\tilde{r}}_*=0)$ is asymptotically stable.

\end{itemize}

\end{proposition} 
\begin{proof}
We propose the new Lyapunov candidate function
\begin{equation*}
V(\xi)=\xi^T Q \xi
\end{equation*}
Since $H$ has eigenvalues with all negative real parts, for any positive definite symmetric matrix $P$, the equation
\begin{equation*}
H^TQ+QH=-P
\end{equation*}
provides a unique solution $Q$. 
Thus, we can express the time derivative of $V$ as
\begin{equation}\label{eq:lyapunov_w_derivative}
\dot{V}=-\xi^TP \xi+2\xi^T QD\left(\mathcal{F}(r)^Tf-\omega\right)
\end{equation}
The second term of (\ref{eq:lyapunov_w_derivative}) becomes
\begin{equation*}
\begin{split}
\lambda^T\left(\mathcal{F}(r)^Tf-\omega\right)&=\lambda^T\left(\mathcal{F}(r)^Tf-\dfrac{\rho}{\|\lambda\|}\lambda\right)\\
&=\lambda^T\mathcal{F}(r)^Tf-\rho\|\lambda\|\\
&\leq \norm{\lambda}\left(\norm{\mathcal{F}(r)^Tf}-\rho\right)\\
\end{split}
\end{equation*}
Thus, by choosing $\rho>\norm{\mathcal{F}(r)^Tf}$, the control signal $\omega$ ensures that $\dot{V}$ is negative definite, and consequently, $\xi\rightarrow 0$. 
\end{proof}

\section{Simulations}\label{sec:sim}


In this section, through two different scenarios, we present simulation results that validate the correct performance of the control strategies presented above.
The simulated MAS consists of six agents moving in a horizontal plane, modeled as rigid bodies. 

The desired topology and orientation are shown in Fig. \ref{fig:Escenario}.
Depending on the scenario to be simulated, the MAS can be represented by the minimally rigid directed graph $\mathcal{G}_1$ with $m=9$ (Fig. \ref{fig:MAS_graph_2}) or a rigid graph $\mathcal{G}_2$ with $m=m_{max}=15$ (Fig. \ref{fig:MAS_graph_4}).
In Fig. \ref{fig:VRS} the VRS is shown.
The trajectory and its corresponding path, shown in Fig. \ref{fig:Ref_trajectory_x} and Fig. \ref{fig:Ref_trajectory_y}, respectively, will be used in both scenarios.
Table \ref{tab:initial_conditions} shows the initial conditions of the mobile bases.

\begin{figure}[H]
\captionsetup[subfigure]{justification=centering}
    \begin{subfigure}[t]{0.49\columnwidth}\centering
        \includegraphics[width=0.9\columnwidth]{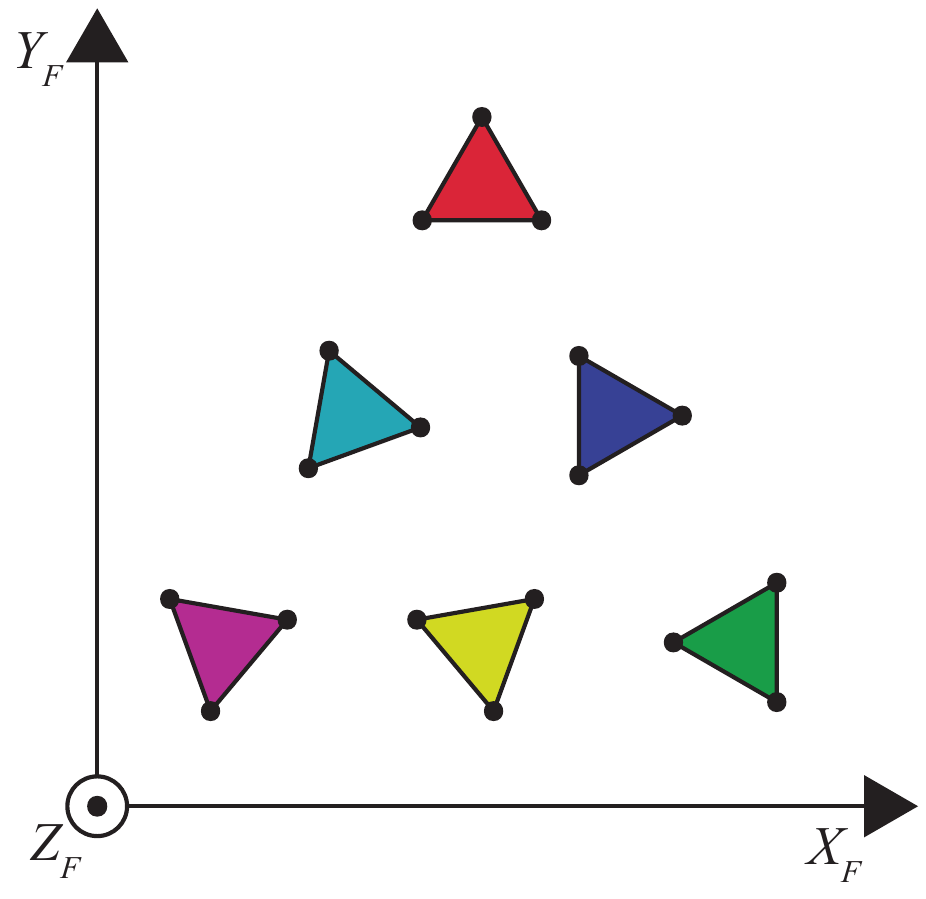}
        \caption{\label{fig:Escenario}Desired topology and orientation.}
    \end{subfigure}
    \begin{subfigure}[t]{0.49\columnwidth}\centering
        \includegraphics[width=0.9\columnwidth]{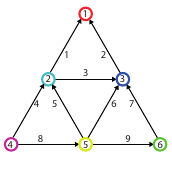}
        \caption{\label{fig:MAS_graph_2}Graph $\mathcal{G}_1$ with $n=6$, $m=9$.}
    \end{subfigure}
    \begin{subfigure}[t]{0.49\columnwidth}\centering
        \includegraphics[width=0.9\columnwidth]{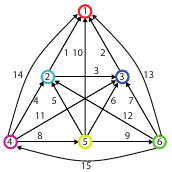}
        \caption{\label{fig:MAS_graph_4}Graph $\mathcal{G}_2$ with $n=6$, $m=15$.}
    \end{subfigure}
    \begin{subfigure}[t]{0.49\columnwidth}\centering
        \includegraphics[width=0.9\columnwidth]{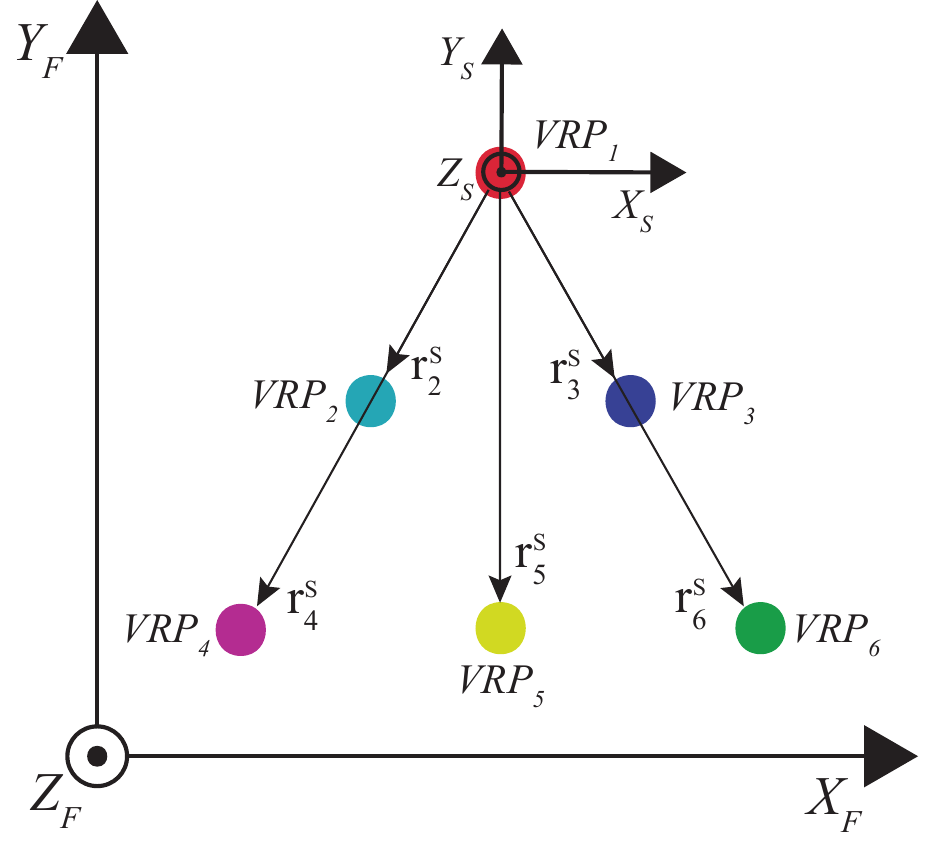}
        \caption{\label{fig:VRS}VRS.}
    \end{subfigure}    
\caption{\label{fig:Ref_formation}Desired formation and graph representations.}
\end{figure}

\begin{figure}[H]
\captionsetup[subfigure]{justification=centering}
    \begin{subfigure}[t]{0.33\columnwidth}\centering
        \includegraphics[width=\columnwidth]{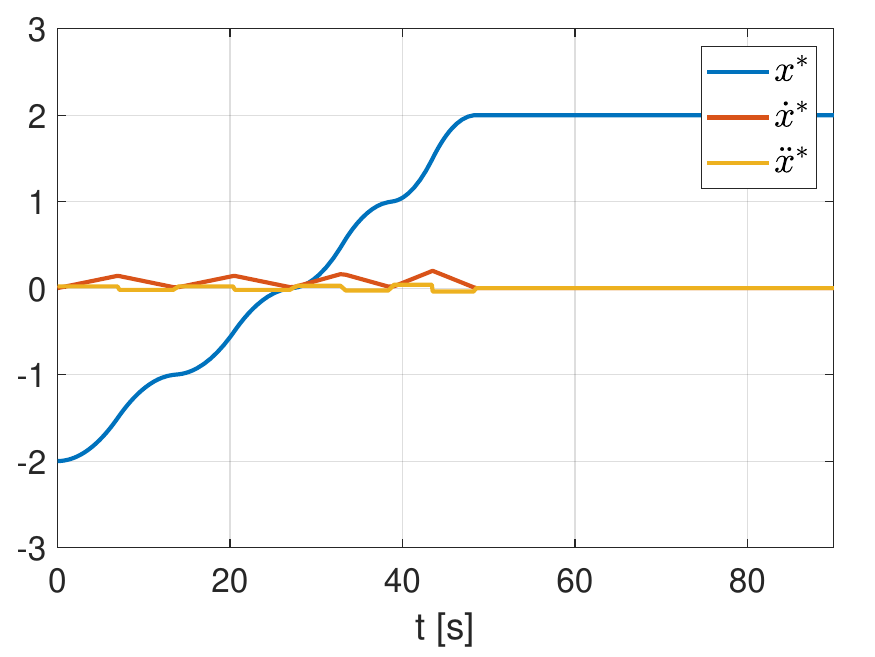}
        \caption{\label{fig:Ref_trajectory_x}Trajectory in $x$.}
    \end{subfigure}%
    \begin{subfigure}[t]{0.33\columnwidth}\centering
        \includegraphics[width=\columnwidth]{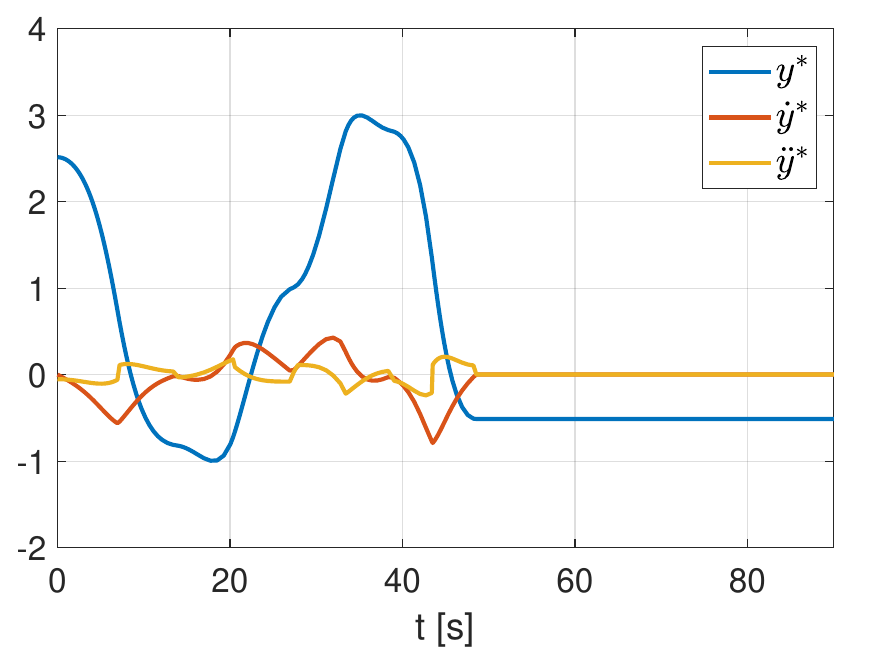}
        \caption{\label{fig:Ref_trajectory_y}Trajectory in $y$.}
    \end{subfigure}%
    \begin{subfigure}[t]{0.33\columnwidth}\centering
        \includegraphics[width=\columnwidth]{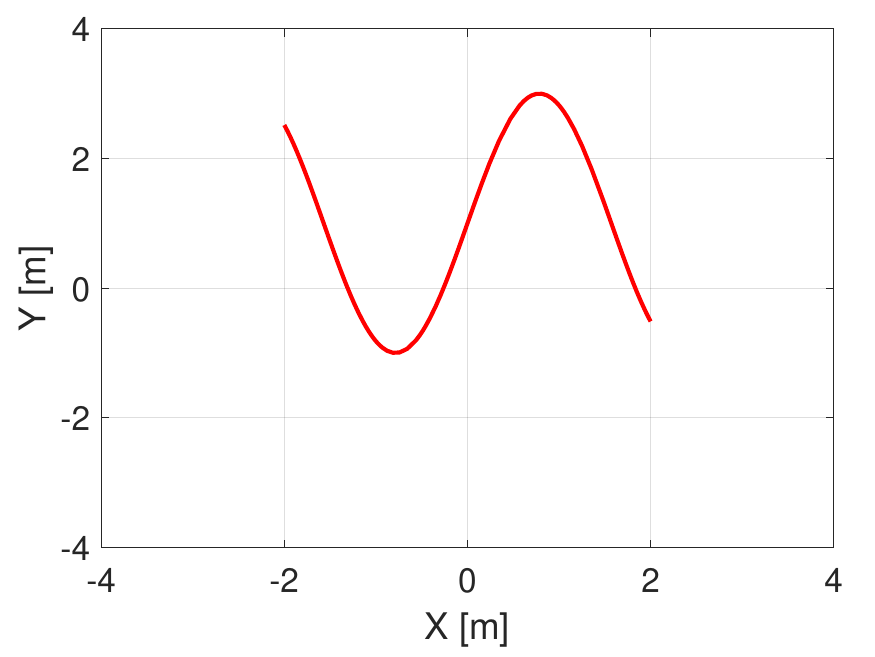}
        \caption{\label{fig:Ref_path}Path in the plane.}
    \end{subfigure}
\caption{\label{fig:Ref_trajectory}Desired trajectory to be used in the different scenarios.}
\end{figure}

\begin{table}[H]
\begin{center}
\caption{Table of initial conditions.}
\label{tab:initial_conditions}
\begin{tabular}{@{}c@{}c@{}c@{}c@{}c@{}c@{}c@{}}
\hline\rule{0pt}{3ex}
 & $1$ & $2$ & $3$ & $4$ & $5$ & $6$ \\\hline \rule{0pt}{3ex}
$q$ & $[-6,5,\frac{\pi}{6}]$ & $[-7,3,0]$ & $[-5,3,\frac{\pi}{4}]$  & $[-7,2,\frac{\pi}{2}]$ & $[-6,1,\frac{3\pi}{2}]$ & $[-6,0,\frac{\pi}{4}]$ \\ \hline\rule{0pt}{3ex}
$\dot{q}$ & $[\frac{1}{4}, \frac{1}{13}, 0]$ & $[\frac{1}{5}, \frac{1}{5}, 0]$ & $[\frac{1}{10}, \frac{3}{20}, 0]$ & $[\frac{1}{50}, \frac{1}{70}, 0]$ & $[\frac{1}{7}, \frac{1}{5}, 0]$ & $[\frac{1}{4}, \frac{1}{30}, 0]$ \rule[-1ex]{0pt}{3.5ex} \\ \hline
\end{tabular}
\end{center}
\end{table}

The Hamiltonian function (\ref{eq:Hamitonian_ECB}) chosen for the ECB$_k$, with $k=1,\cdots,m$, is
\begin{equation*}
H_{ck}(\tilde{z}_k)=\begin{cases}\left(\tilde{z}_k+\frac{{z_k^*}^3}{2\left(\tilde{z}_k+z_k^*\right)^2}-\frac{z_k^*}{2}\right)\mu_R\quad \textrm{if}\;\; \tilde{z}_k<0\\[1.5em]
\frac{3}{2z_k^*}\mu_R\tilde{z}_k^2\quad \textrm{if}\;\; \tilde{z}_k\geq 0
\end{cases}
\end{equation*}
which possesses a continuous derivative throughout its domain.

The values of the ECBs parameters are $\mu_R=10$, $B=20I_9$ and the desired lengths of all distances are set to $z_k^*=1.5$, with $k=1,\cdots,m$.
The mobile bases parameters are:  $m_b=1.5$ Kg, $L=0.2$m, $r=0.024$m, $I_b=0.0216$ Kgm$^2$, and $J_r=12.24$e$^{-6}$ Kgm$^2$.

\subsection{Scenario 1: Leader Agent Control}
The graph used is $\mathcal{G}_1$ (Fig. \ref{fig:MAS_graph_2}).
Agent $i=1$ is chosen as the leader, and edge $k=1$ is utilized to control the angle (see Fig. \ref{fig:MAS_graph_2}). 
The evolution of the first two components of the reference vector $a^*=\left[x_1^*,\;y_1^*,\;\alpha_1^*\right]^T$ are shown in Fig. \ref{fig:Ref_trajectory_x} and Fig. \ref{fig:Ref_trajectory_y}, respectively. 
The third component, the desired angle, remains constant and is equal to $\alpha_1^*=60°$, which means that the topology moves in the plane with the orientation shown in Fig. \ref{fig:Escenario}.
The parameters of the control law (\ref{eq:v_leaderagent}) are $K_{Pa}=20\I_3,\, K_{Ia}=20\I_3 $.

In Fig. \ref{fig:Sim_D_BM}, it can be verified that the equilibrium point ($p_*=0,\tilde{z}_*=0 ,\tilde{a}_*=0$) is stable as the MAS comes to a stop (Fig. \ref{fig:Sim_D_BM_x} and \ref{fig:Sim_D_BM_y}), distances between agents reach their desired values (Fig. \ref{fig:Sim_D_BM_Tz}), and the error $\tilde{a}$ tends to zero (Fig. \ref{fig:Sim_D_BM_Ta}). Observing the evolution of this last figure confirms that the \textit{Leader Agent Control} stabilizes the MAS at the equilibrium point when the reference $a^*$ is constant (starting from $t=50s$).

\begin{figure}[b!]
    \begin{subfigure}[t]{0.49\columnwidth}\centering
        \includegraphics[width=\columnwidth]{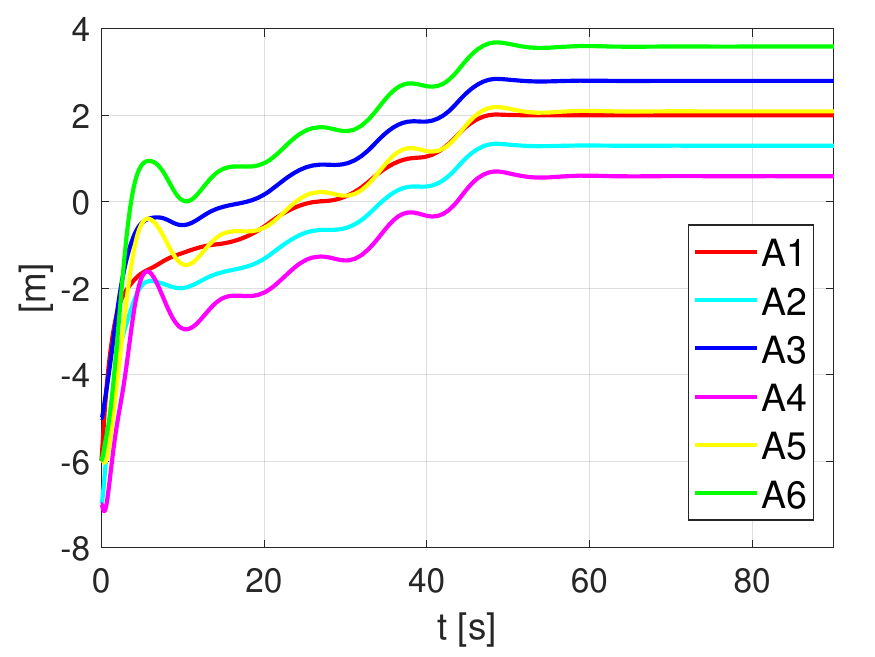}
        \caption{\label{fig:Sim_D_BM_x}$x$-coordinate.}
    \end{subfigure}
    \begin{subfigure}[t]{0.49\columnwidth}\centering
        \includegraphics[width=\columnwidth]{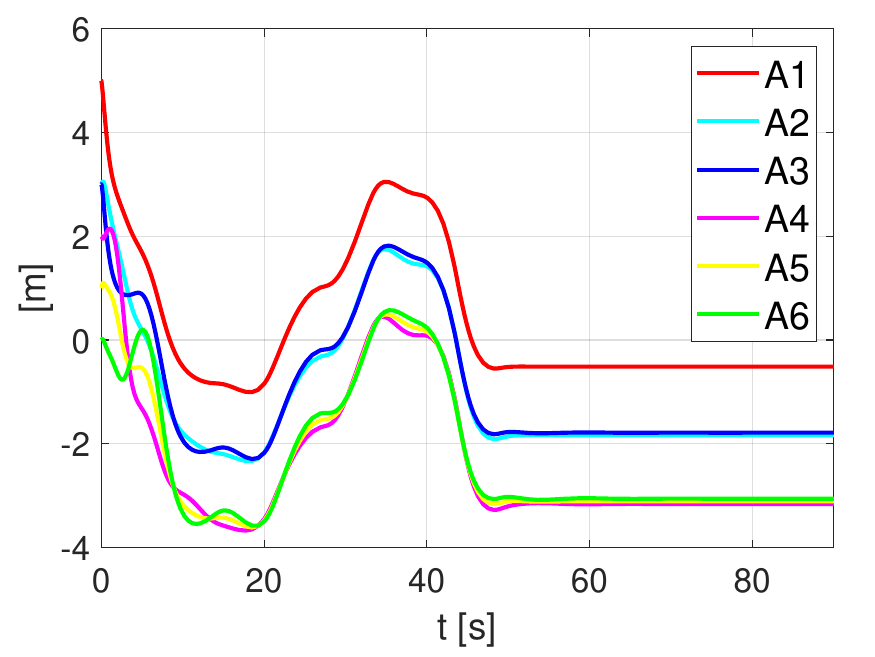}
        \caption{\label{fig:Sim_D_BM_y}$y$-coordinate.}
    \end{subfigure}
    \begin{subfigure}[t]{0.49\columnwidth}\centering
        \includegraphics[width=\columnwidth]{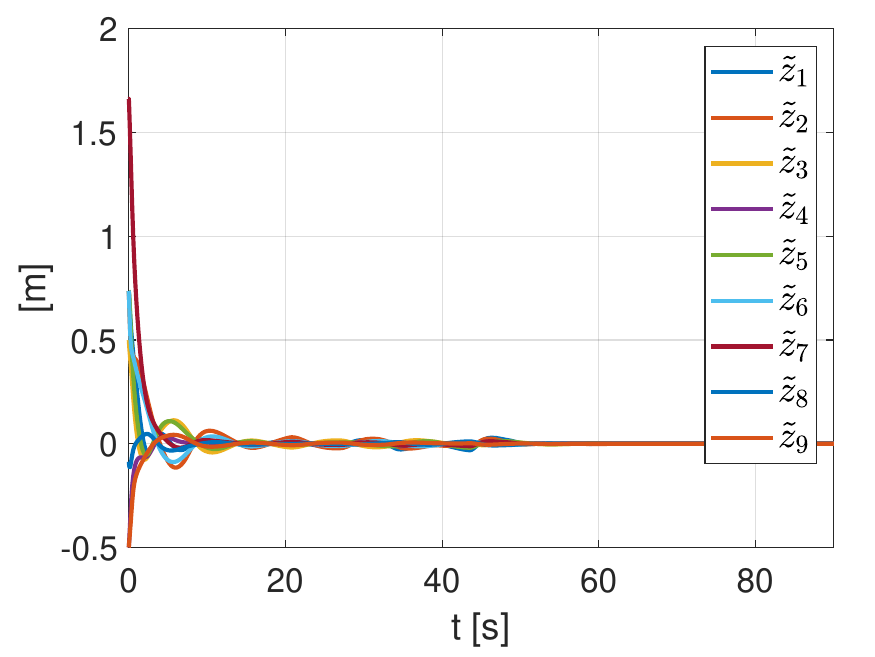}
        \caption{\label{fig:Sim_D_BM_Tz}Distance error $\tilde{z}$.}
    \end{subfigure}
    \begin{subfigure}[t]{0.49\columnwidth}\centering
        \includegraphics[width=\columnwidth]{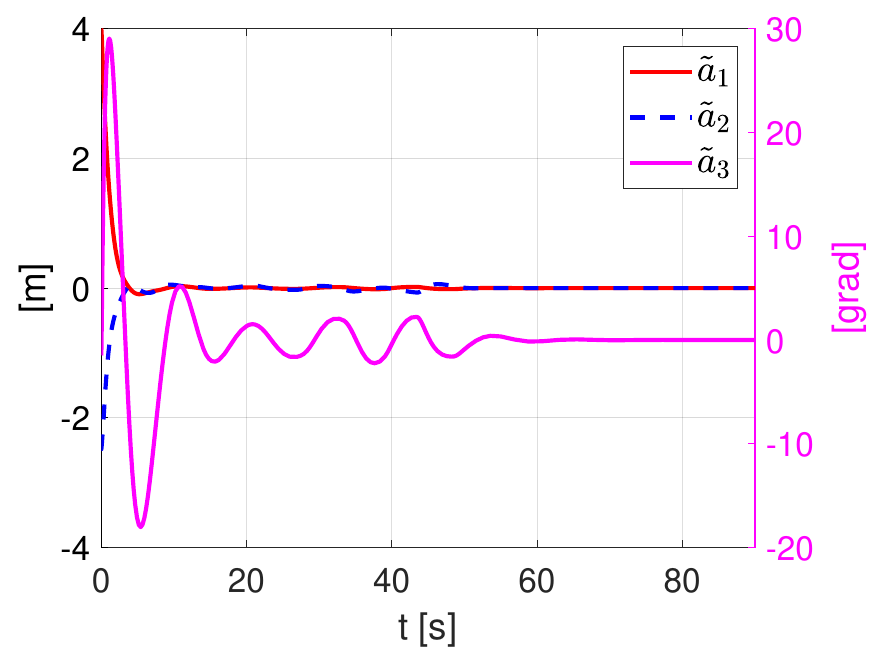}
        \caption{\label{fig:Sim_D_BM_Ta}Error $\tilde{a}$.}
    \end{subfigure}
\caption{\label{fig:Sim_D_BM}\textit{Leader Agent Control} strategy.}
\end{figure}

On the other hand, the evolutions in the plane at different time instants, as depicted in Fig. \ref{fig:Sim_D_BM_plane}, illustrate that the agents form and follow the trajectory with the desired topology and orientation.

\begin{figure}[t!]
    \begin{subfigure}[t]{0.49\columnwidth}
        \includegraphics[width=\columnwidth]{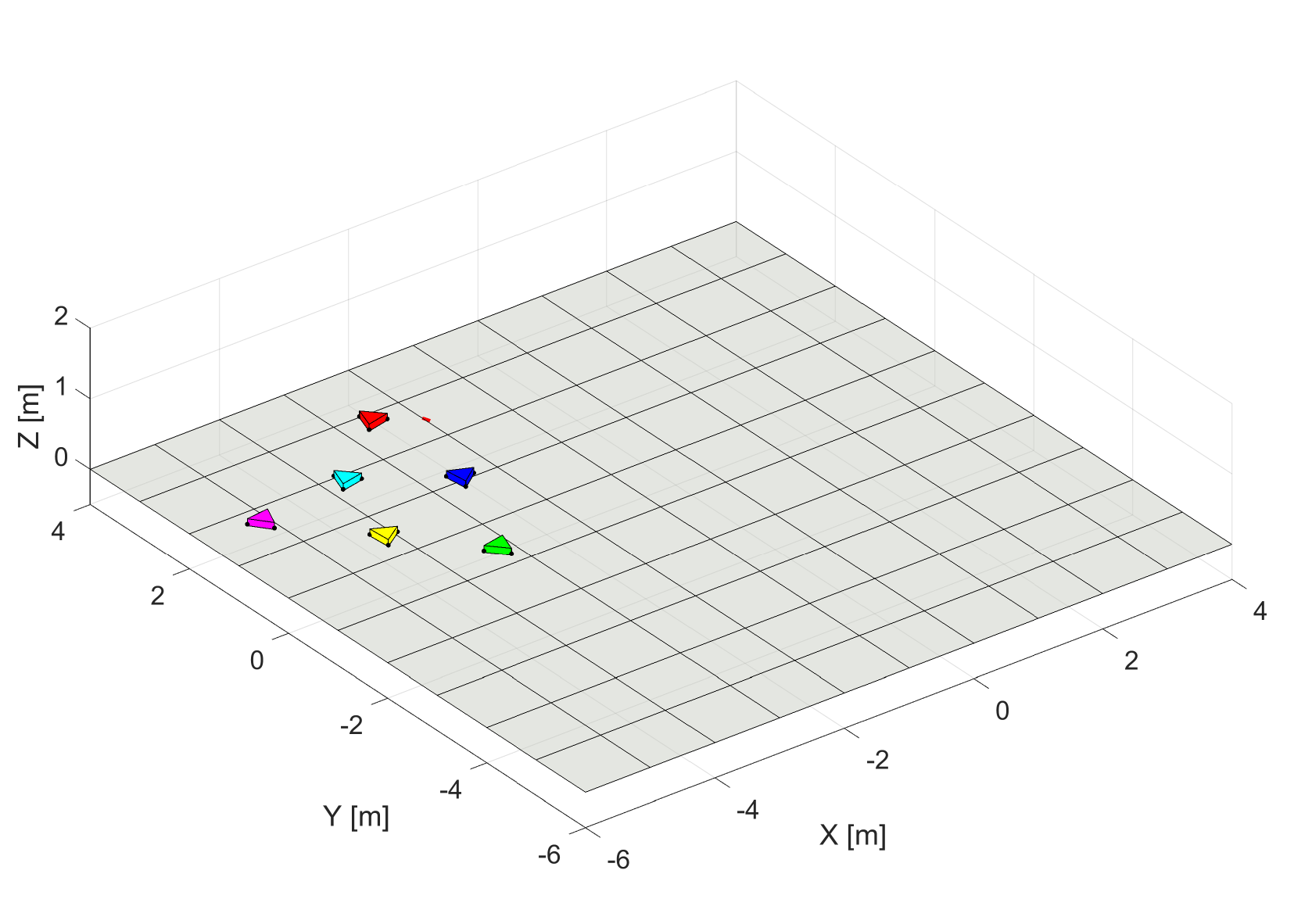}
        \caption{$t=2.5s$.}
    \end{subfigure}
    \begin{subfigure}[t]{0.49\columnwidth}
        \includegraphics[width=\columnwidth]{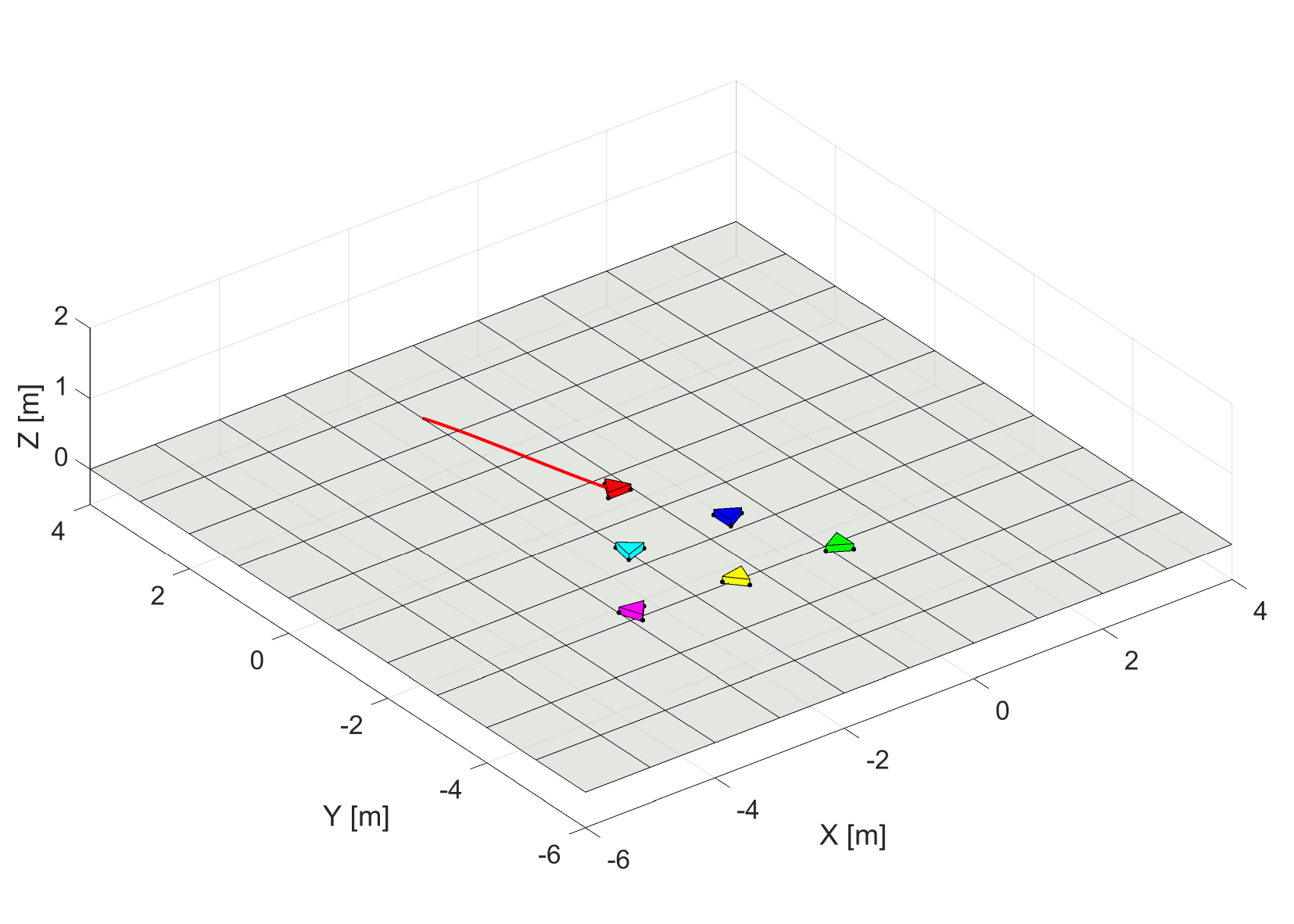}
        \caption{$t=10s$.}
    \end{subfigure}
    \begin{subfigure}[t]{0.49\columnwidth}
        \includegraphics[width=\columnwidth]{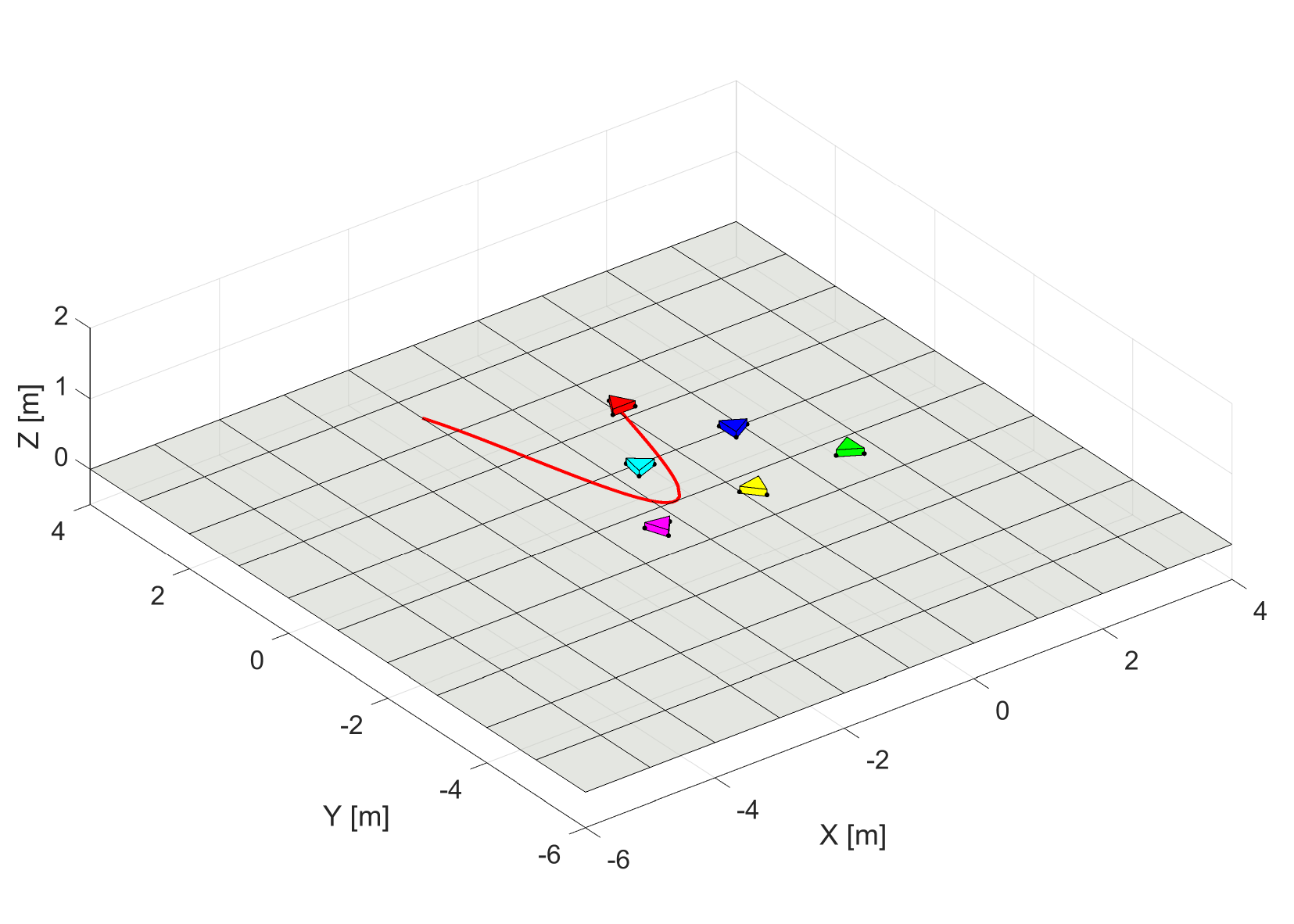}
        \caption{$t=30s$.}
    \end{subfigure} 
    \begin{subfigure}[t]{0.49\columnwidth}
        \includegraphics[width=\columnwidth]{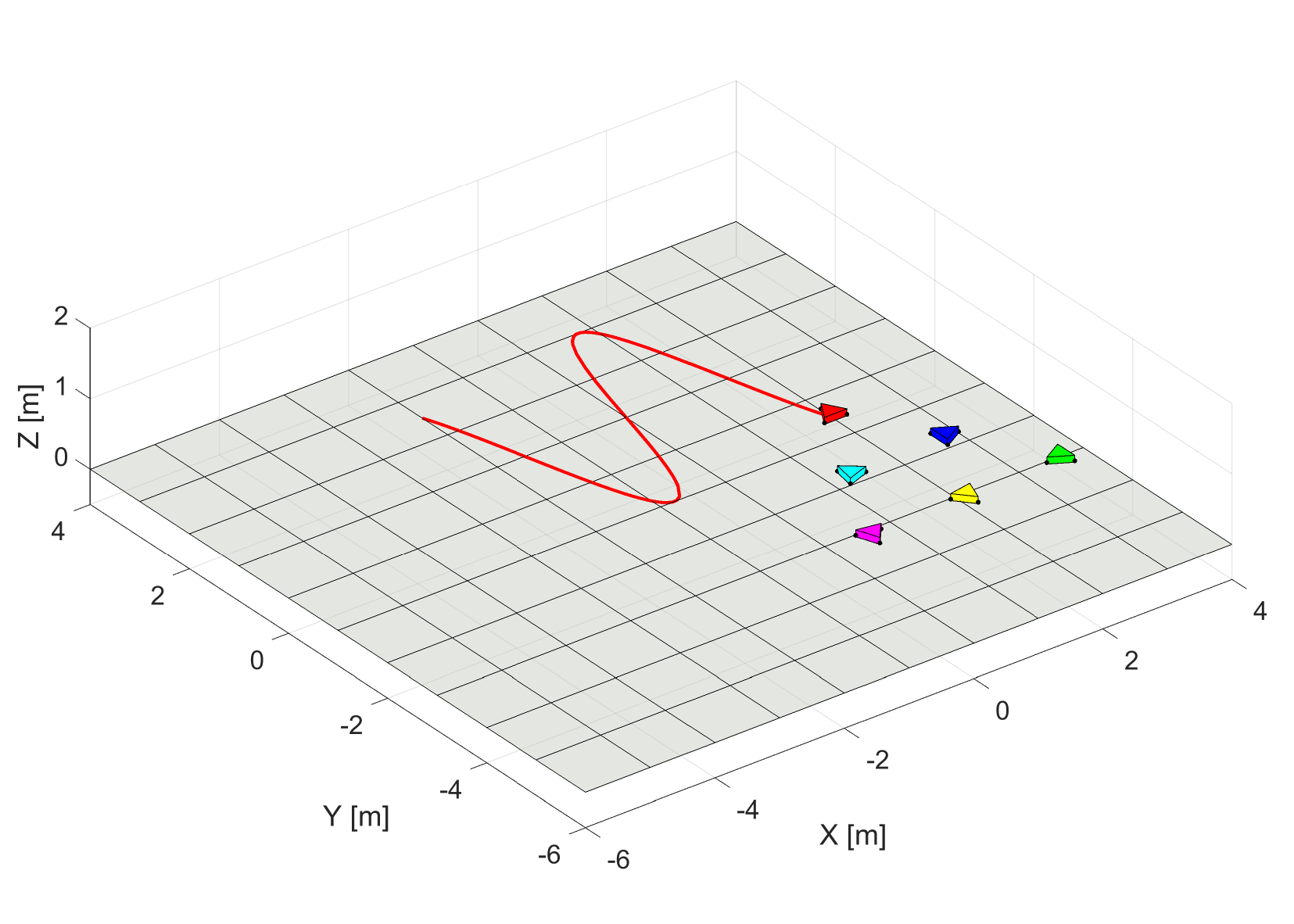}
        \caption{$t=60s$.}
    \end{subfigure}  
    \caption{\label{fig:Sim_D_BM_plane}\textit{Leader Agent Control} strategy - movements in the plane - Graph $\mathcal{G}_1$.}  
\end{figure}

\begin{remark}\label{re:desventajas_lider}
In this formation control scheme, there may be cases where, even while respecting the distance restrictions imposed by the controller, the MAS cannot achieve the desired topology or orientation. Specifically, it may happen that agents form in the desired topology but mirrored with respect to a local axis of symmetry, or that, for a certain restriction on the distance between agents, there are different possible positions for them.
This is because the essence of the controller lies in the concept of the \textit{Formation Matrix}, and as demonstrated, for it to be of full rank, it is necessary for the graph characterizing the MAS to be minimally rigid. Therefore, the number of edges will be less than the set of distances between agents that need to be controlled. 
Moreover, in the case of a topological change, caused for example by a communication failure or a dynamic reconfiguration, new distances must be controlled in order to preserve the minimal rigidity of the graph.
\end{remark}

\subsection{Scenario 2: Position Control}

In this control scheme, as emphasized in Remark \ref{rem:ECB_control_posicion}, the graph chosen is that of Fig. \ref{fig:MAS_graph_4} with $m=m_{max}=15$. 
The VRS is shown in Fig. \ref{fig:VRS}. The origin of the reference frame $S$ follows the path shown in Fig. \ref{fig:Ref_path}, with a constant orientation of  $\varphi_S = 0$.
The parameters of the control law (\ref{eq:v_positioncontrol}) are $K_{Ps}=5\I_{12} ,\, K_{Ds}=10\I_{12}$.

Simulation results of the \textit{Position Control} without $\omega$ ($\rho=0$ in (\ref{eq:omega})) are shown in the plots of Fig. \ref{fig:Sim_G}, where it can be verified that the system tends towards an equilibrium that is not the desired one. 
This is evident in the norm of the position error $|\tilde{r}|$ and the distance error $\tilde{z}$, from Fig. \ref{fig:Sim_G_NormTr} and Fig. \ref{fig:Sim_G_Tz}, respectively. Both do not converge to zero but to a value such that the agents do not form the desired topology. 
On the other hand, collision avoidance is successfully achieved, as seen in Fig. \ref{fig:SIM_G_z}, with distances between agents greater than one meter.

\begin{figure}[b!]
\captionsetup[subfigure]{justification=centering}
    \begin{subfigure}[t]{0.33\columnwidth}\centering
        \includegraphics[width=\columnwidth]{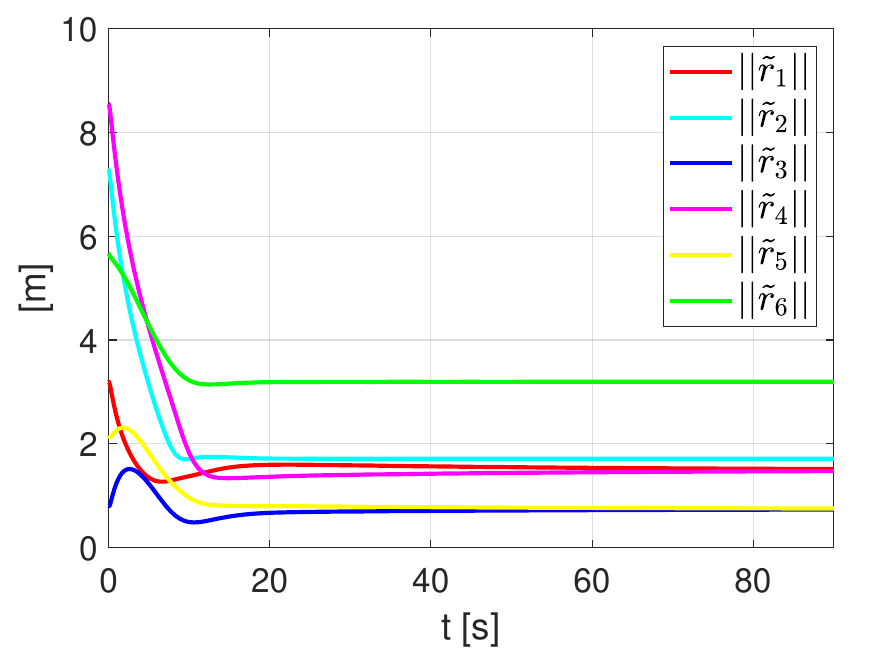}
        \caption{\label{fig:Sim_G_NormTr}Position error norm $\|\tilde{r}\|$.}
    \end{subfigure}
    \begin{subfigure}[t]{0.32\columnwidth}\centering
        \includegraphics[width=\columnwidth]{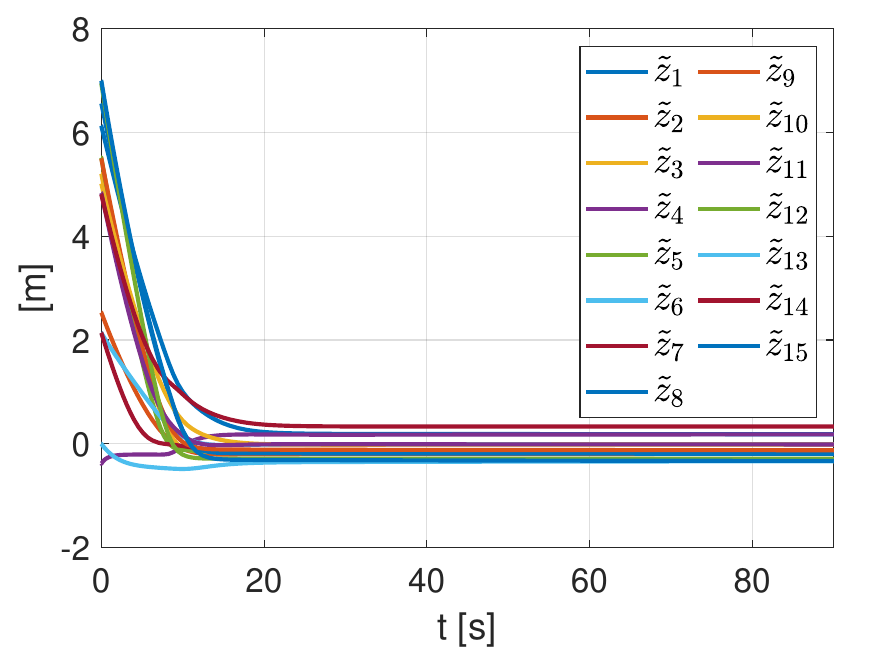}
        \caption{\label{fig:Sim_G_Tz}Distance error $\tilde{z}$.}
    \end{subfigure}
    \begin{subfigure}[t]{0.32\columnwidth}\centering
        \includegraphics[width=\columnwidth]{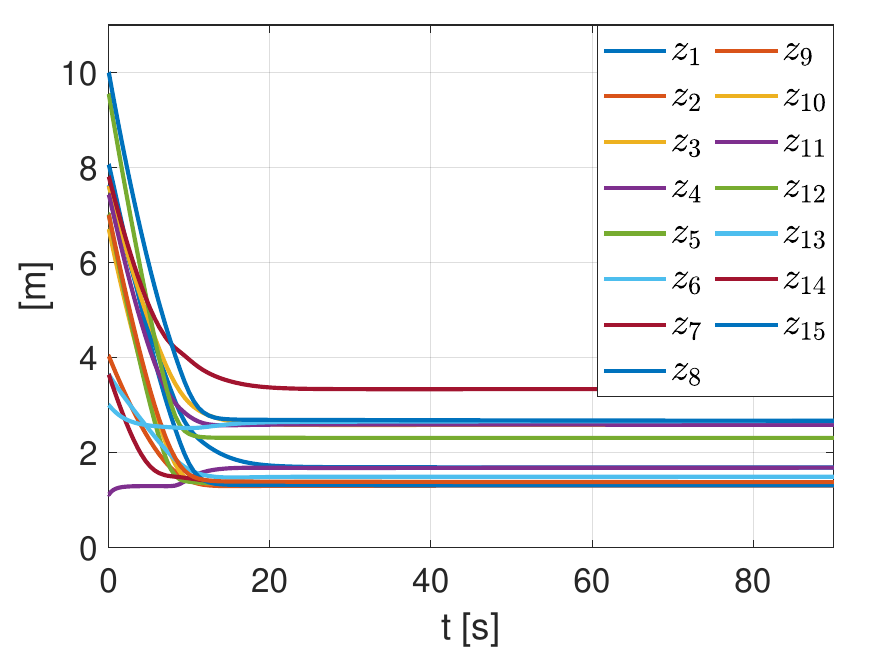}
        \caption{\label{fig:SIM_G_z}Distances $z$.}
    \end{subfigure}
\caption{\label{fig:Sim_G}\textit{Position Control} strategy without $\omega$.}
\end{figure}

The incorporation of the control signal $\omega$ ($\rho=50$ in (\ref{eq:omega})) effectively ensures that the trajectory error $\tilde{r}$ and the distance errors $\tilde{z}$ tend to zero, as can be observed in Fig. \ref{fig:Sim_H_NormTr} and Fig. \ref{fig:Sim_H_Tz}, respectively. This confirms that the control law $\omega$ ensures asymptotic stability of the desired equilibrium point $(\tilde{r}_*=0,\dot{\tilde{r}}_*=0,\tilde{z}_*=0)$.
Additionally, in Fig. \ref{fig:Sim_H_z}, it is observed that the distances are maintained above 0.9 meters throughout the simulation, ensuring that the agents never collide with each other.

The evolution of the control signal $\omega$ is shown in Fig. \ref{fig:Sim_H_w}, where it can be seen that, once the agents start following their desired trajectories, it becomes zero.

\begin{figure}[t!]
\captionsetup[subfigure]{justification=centering}
    \begin{subfigure}[t]{0.49\columnwidth}\centering
        \includegraphics[width=\columnwidth]{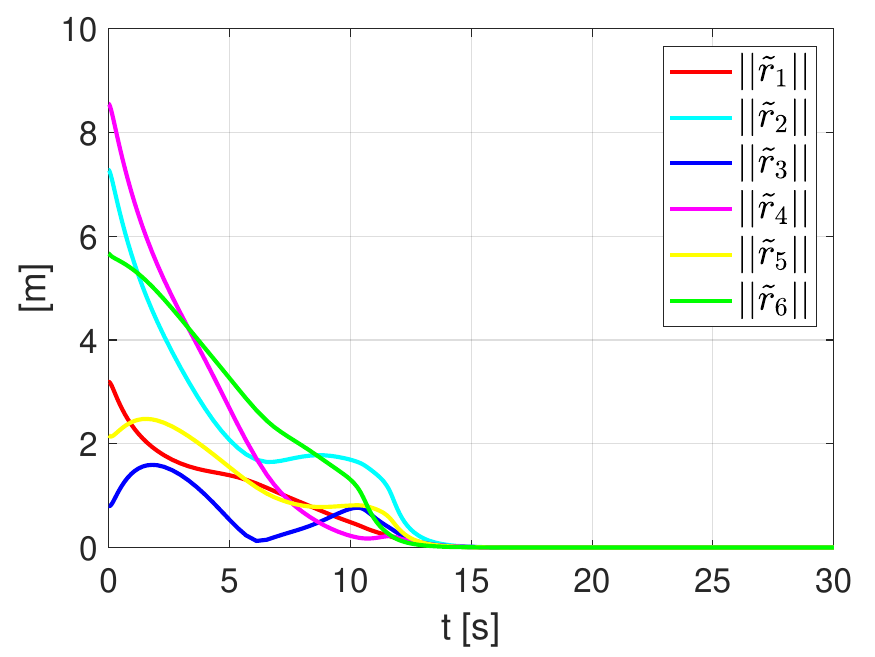}
        \caption{\label{fig:Sim_H_NormTr}Position error norm $\|\tilde{r}\|$.}
    \end{subfigure}%
    \begin{subfigure}[t]{0.49\columnwidth}\centering
        \includegraphics[width=\columnwidth]{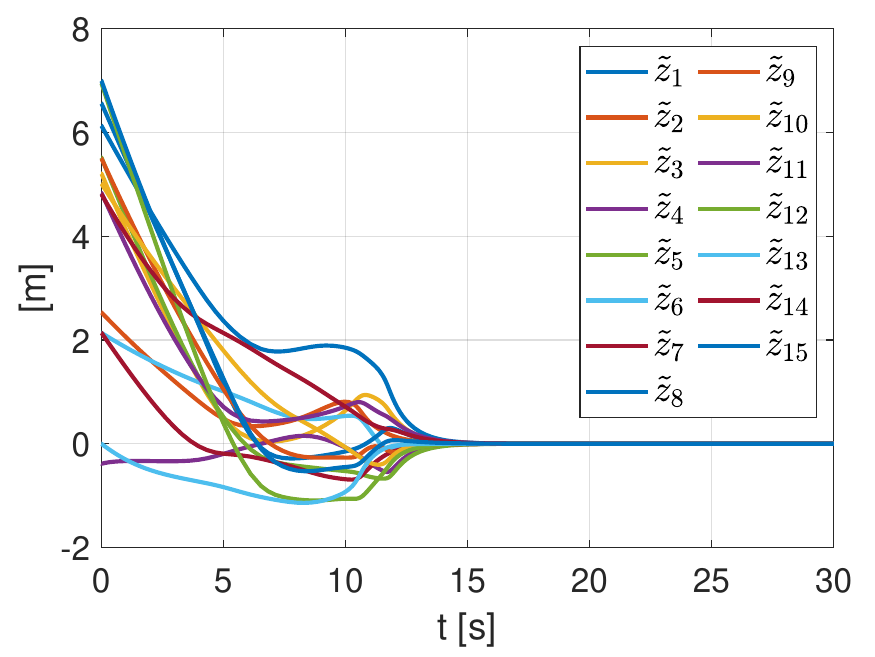}
        \caption{\label{fig:Sim_H_Tz}Distance error $\tilde{z}$.}
    \end{subfigure}
    \begin{subfigure}[t]{0.49\columnwidth}\centering
        \includegraphics[width=\columnwidth]{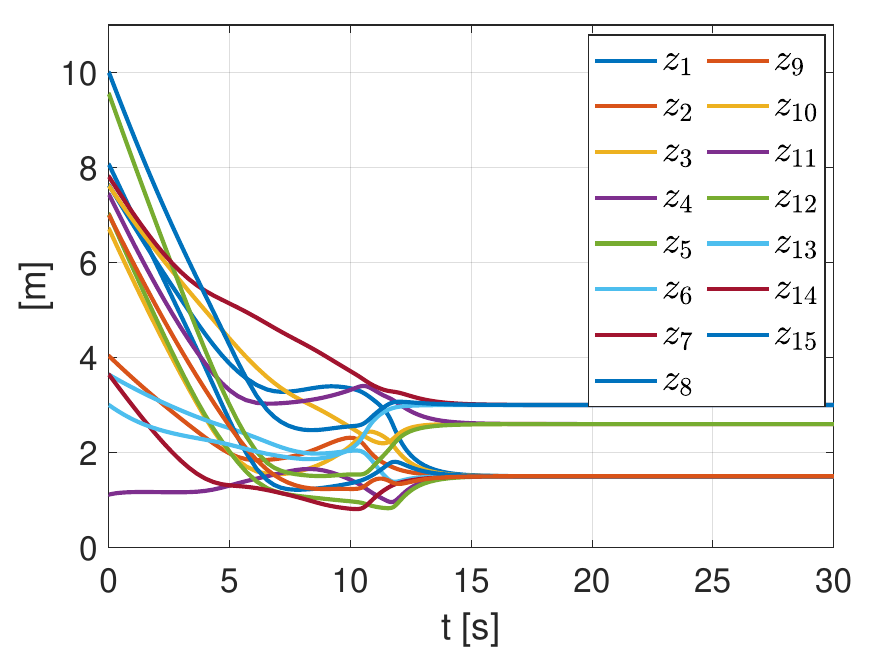}
        \caption{\label{fig:Sim_H_z}Distances $z$.}
    \end{subfigure}
    \begin{subfigure}[t]{0.49\columnwidth}\centering
        \includegraphics[width=\columnwidth]{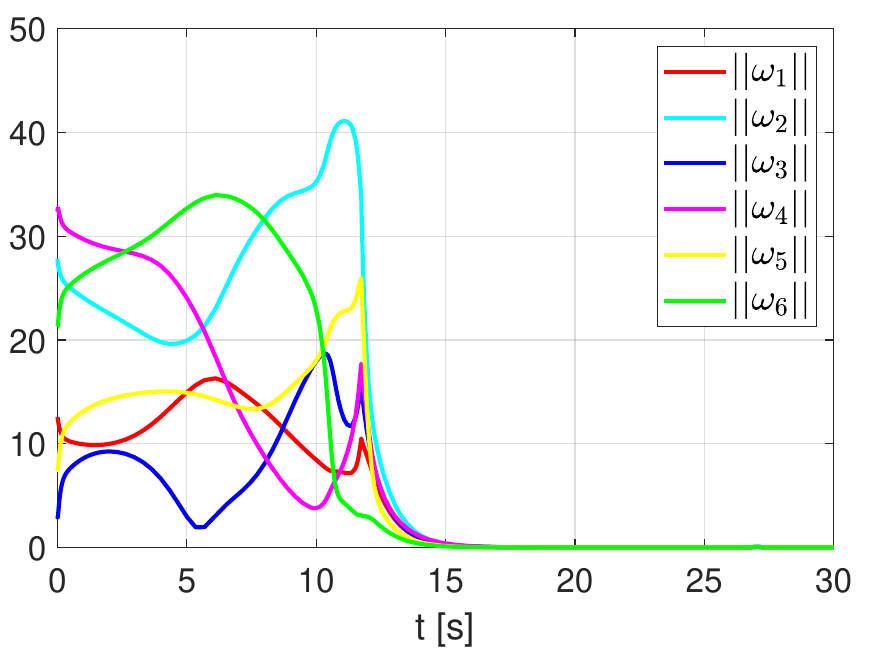}
        \caption{\label{fig:Sim_H_w}Control signal norm $\|\omega\|$.}
    \end{subfigure}
\caption{\label{fig:Sim_H}\textit{Position Control} strategy with $\omega$ (plotted until $t=30s$).}
\end{figure}

\begin{remark}\label{rem:chattering}
In this particular case, the convergence to zero of the trajectory error is ensured by the position control. 
That is, the incorporation of the control law $\omega$ is not intended to reject disturbances neither parametric uncertainties; instead, it is aimed at ensuring convergence to the desired equilibrium point in closed-loop. 
In this way, the control signal does not oscillate in the phenomenon known as chattering but tends towards $\omega=0$.
\end{remark}

In the snapshots of the evolution, shown in Fig. \ref{fig:Sim_H_2}, the action of the control law $\omega$ to reposition the agents in the desired topology is graphically illustrated. By the time $t=10s$, the agents have formed the desired topology and orientation, with agent $i=1$ following the reference trajectory.

\begin{figure}[t!]
    \begin{subfigure}[t]{0.49\columnwidth}
        \includegraphics[width=\columnwidth]{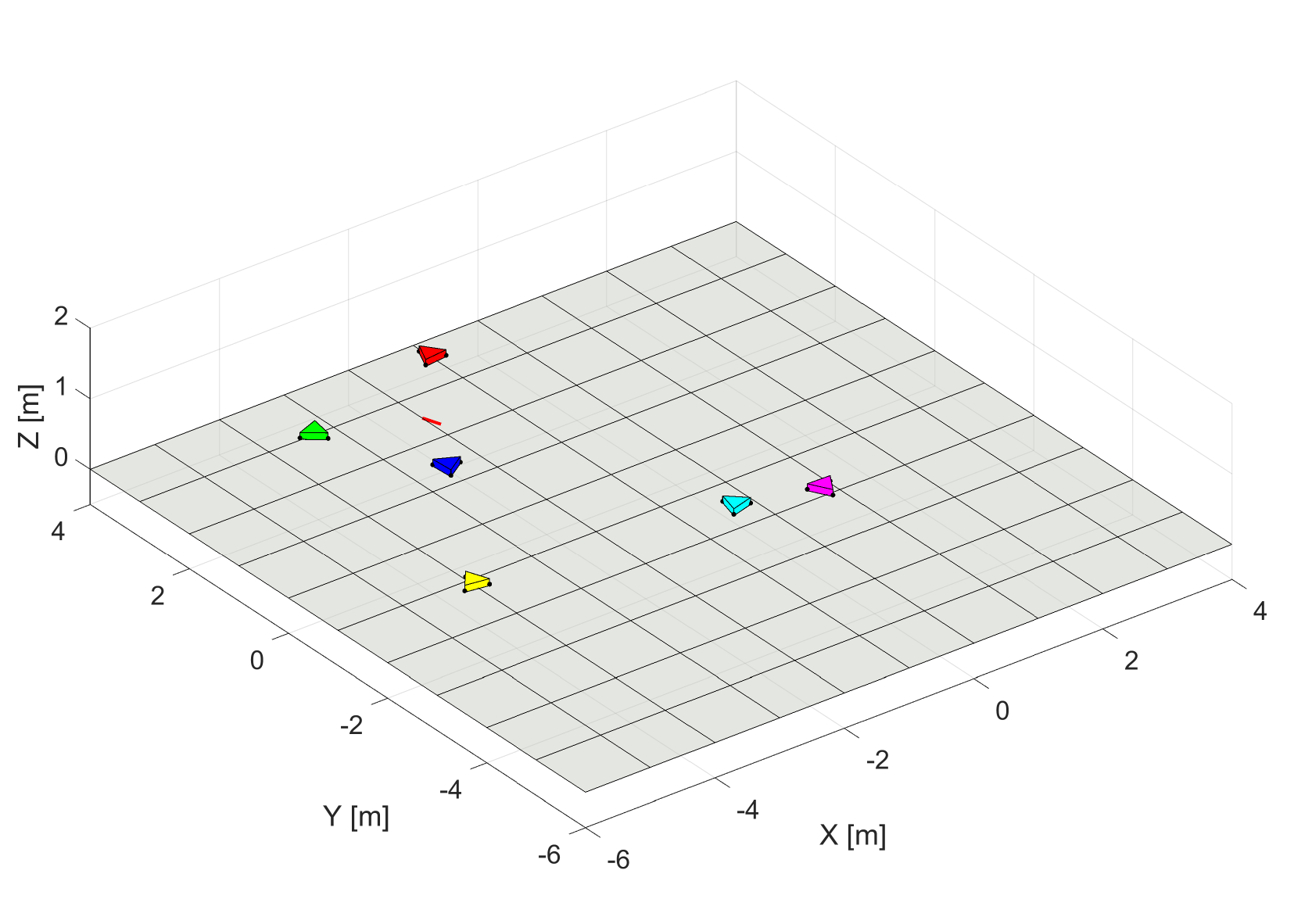}
        \caption{$t=4s$.}
    \end{subfigure}
    \begin{subfigure}[t]{0.49\columnwidth}
        \includegraphics[width=\columnwidth]{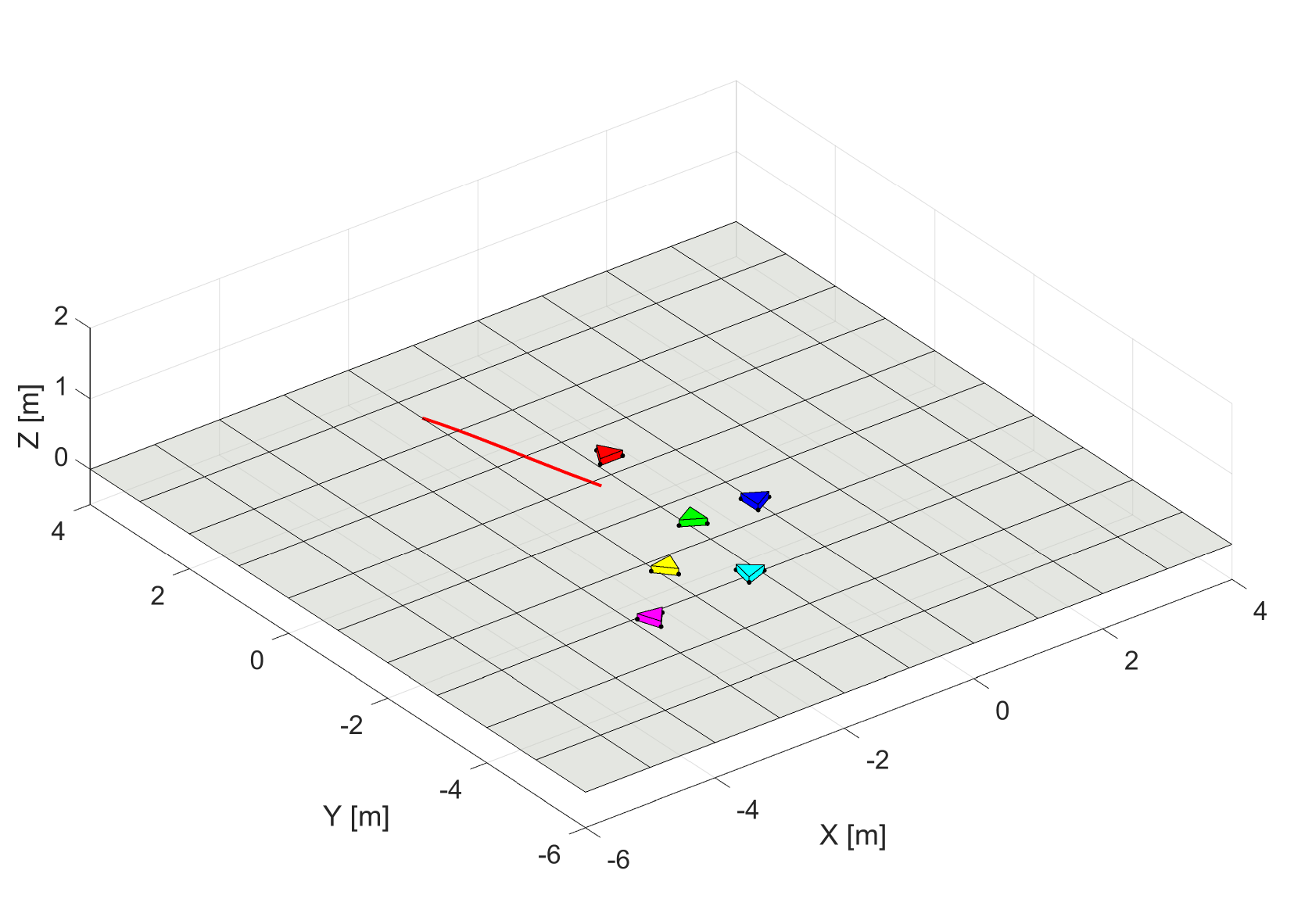}
        \caption{$t=10s$.}
    \end{subfigure}
    \begin{subfigure}[t]{0.49\columnwidth}
        \includegraphics[width=\columnwidth]{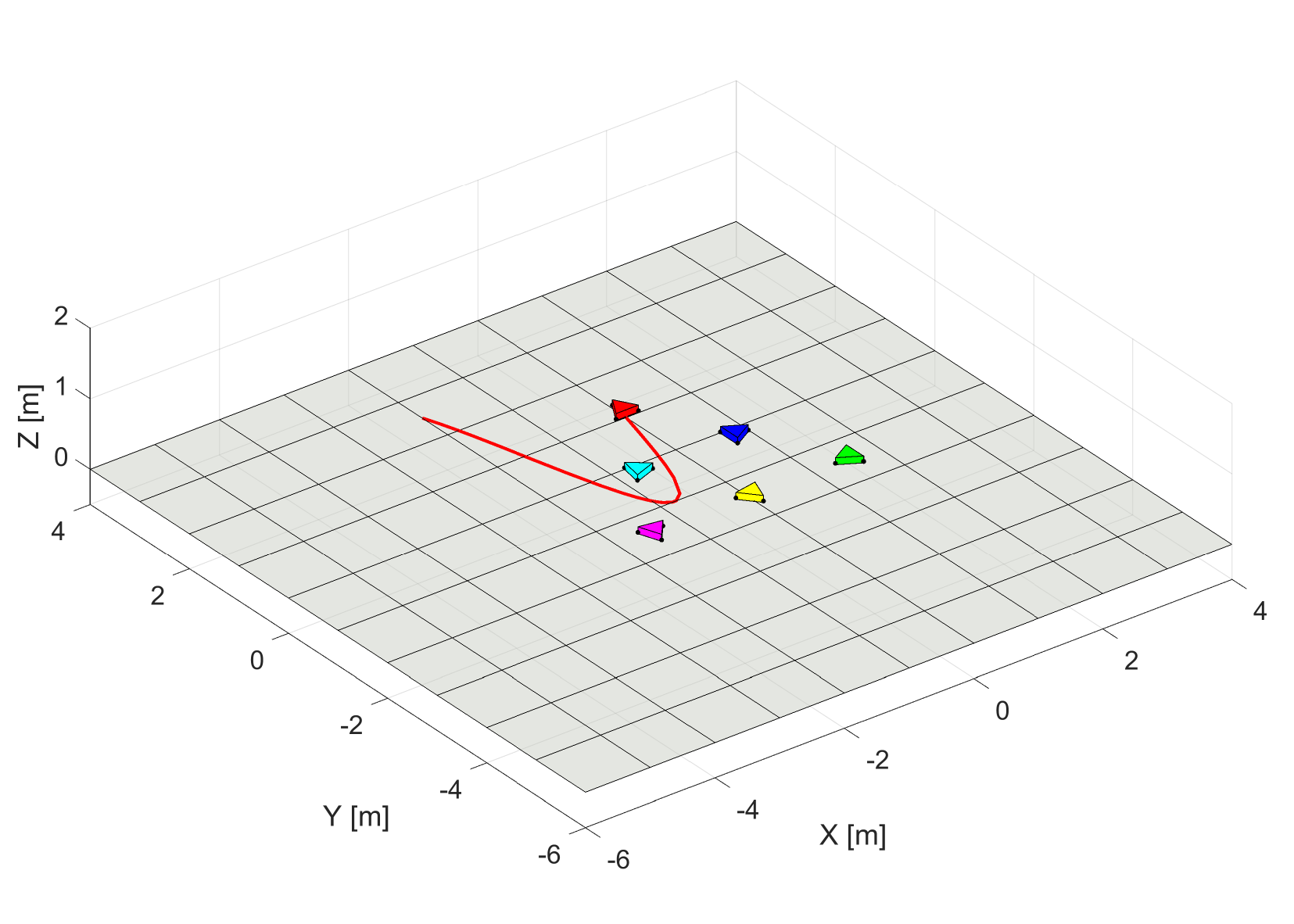}
        \caption{$t=30s$.}
    \end{subfigure} 
    \begin{subfigure}[t]{0.49\columnwidth}
        \includegraphics[width=\columnwidth]{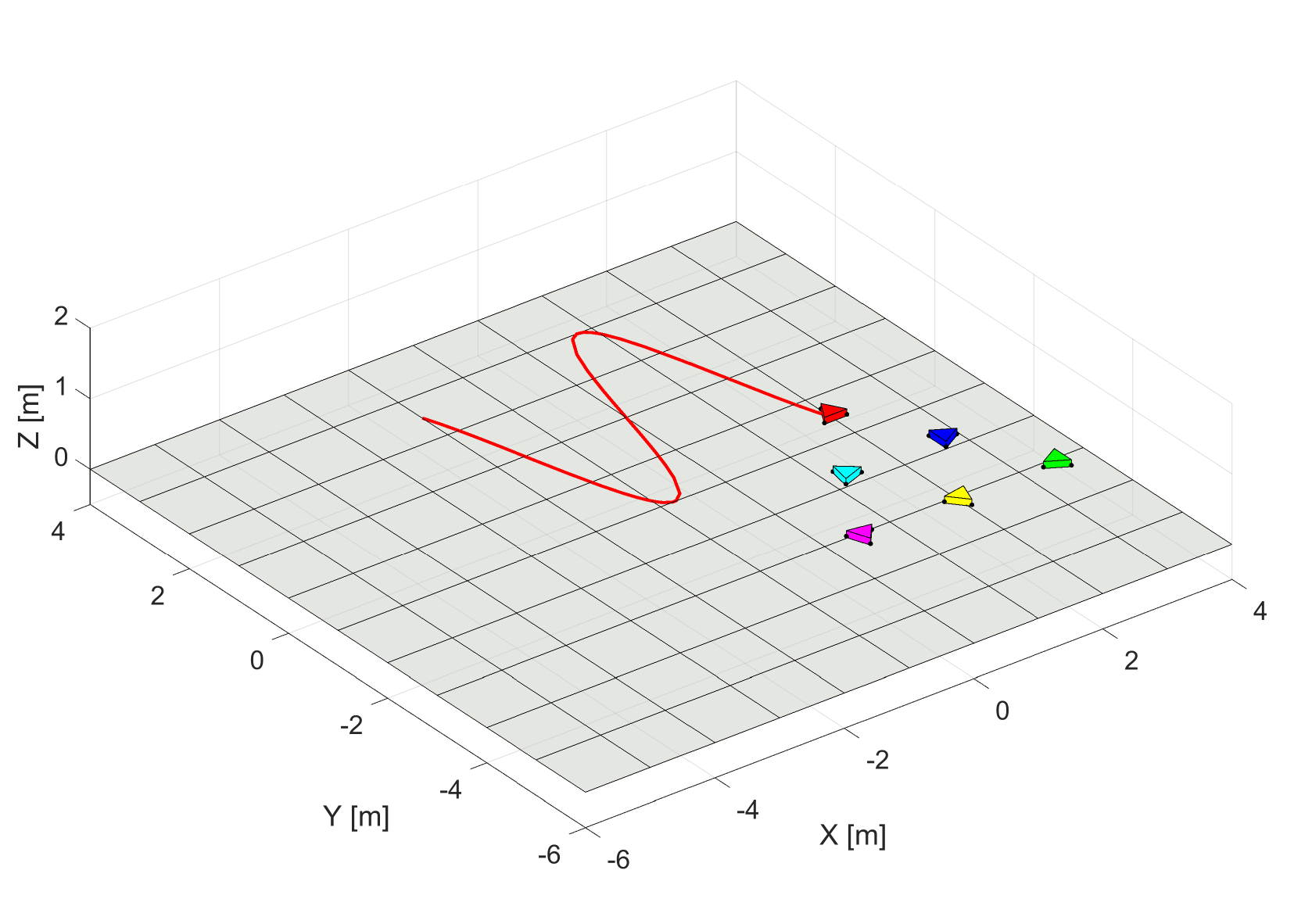}
        \caption{$t=80s$.}
    \end{subfigure}  
    \caption{\label{fig:Sim_H_2}\textit{Position Control} strategy - movements in the plane - Graph $\mathcal{G}_2$.}  
\end{figure}

\subsection{Discussions}

Some advantages of the control strategies that we can mention are:
\begin{itemize}
\item In the \textit{Leader Agent Control}, both formation control and collision prevention are carried out through a single controller, which also has a physical interpretation.
\item Only the \textit{Position Control} ensures the orientation of the topology is always as desired, with no risks of collisions between any of the agents.
\end{itemize}

Among the disadvantages are:
\begin{itemize}
\item In the \textit{Leader Agent Control}, for the matrix $\mathcal{T}(r)$ to be full rank, the graph must be minimally rigid, resulting in the formation being characterized by a graph of $m=9$ edges. 
This limitation does not avoid the collisions between agents whose distances are not represented by the graph $\mathcal{G}_1$.
\item \textit{Position Control} lacks a physical interpretation, and the control is more focused on each individual agent rather than the formation itself.
\end{itemize}

Formation controls can be classified based on the controlled variable whether it is position, displacement, or distance \cite{oh2015survey}.
Clearly, \textit{Position Control} classifies as a position-based formation control, while the \textit{Leader Agent Control} scheme result in a hybrid between distance and displacement control.
This is because, even though this strategy actively controls the distance, calculating the control signals involves measuring global displacements between agents for the computation of the \textit{Formation Matrix}.
Furthermore, since it is necessary to know the position of the leader agent and its neighbor, it could also fall into the classification of position-based control.

Both formation schemes addressed can also be differentiated based on their interpretation. 
The \textit{Leader Agent Control} strategy is based on basic elements of mechanics (springs and dampers) and treats the MAS as a semi-rigid body composed of interconnected subsystems (via the ECBs). Therefore, its results have a physical interpretation.
On the other hand, the \textit{Position Control} introduces a VRS, where the control law expressions emerge from the geometry of the structure.  Hence, its interpretation is purely mathematical.

In a physical scenario, the proposed strategy must be implemented in a centralized scheme, since computing the \textit{Formation Matrix} requires the global positions of the agents. These positions can be obtained through the fusion of camera measurements, IMU data, and odometry.

Concerning the methodology used in this work, we emphasize that the BG technique has been instrumental in the derivation of the formation control law at several instances. First, it allowed for a straightforward representation of the energy phenomena (storage, dissipation) and the physical interconnection and power exchange structure of the proposed virtual physically-inspired solution. Second, the causal relationships provided by the SCAP (Sequential-Causality Assignment Procedure) \cite{karnopp2012system}, allowed for the derivation of a dynamic control law a-priori solving the problem. Moreover, these very same causal relationships reveal, in conjunction with the \textit{Formation Matrix}, the existence of an algebraic dependence between the MAS and the controller state variables, fact that immediately suggest this dependence being a Casimir function and, thus, the possibility of rewriting the controller means of a static state-feedback law. This is immediately verified and proved via the straightforward reformulation of the problem into the CbI-form in the pHS-framework.

\section{Conclusion}\label{sec:Con}

In this work we presented two control strategies based on the novel concept of  \textit{Formation Matrix}.
We have demonstrated that for minimally rigid graphs, the \textit{Formation Matrix} has full rank, a key feature to ensure the stability of the equilibrium point when implementing the \textit{Leader Agent Control} strategy.
We also demonstrated that the physical approach with ECBs, combined with a position-based formation control and sliding mode control, worked correctly while avoiding high-frequency commutation.
Regarding stability, in both schemes, stability of the equilibrium point was demonstrated in the Lyapunov sense. In the specific case of the \textit{Position Control}, the objective of the control law $\omega$ was to ensure convergence to the desired equilibrium point.
Future work is focused on extending the results to differential mobile bases and studying the limitations of distance-based formation control strategies in ensuring the desired orientation and topology of the formation.
Future work aims also to extend the \textit{Formation Matrix} concept to three-dimensional scenarios, e.g., for unmanned aerial vehicles, by measuring their positions in space.

\section*{Acknowledgment}
The authors wish to thank the National University of Rosario for funding this work through the projects PPCT-80020220600083UR and PID-80020190300098UR.

\bibliographystyle{IEEEtran}
\bibliography{IEEEabrv,Biblio_paper}

\end{document}